\documentclass[a4paper,11pt]{llncs}

\usepackage{microtype}

\usepackage{wrapfig}
\usepackage{microtype}
\usepackage{graphicx}
\usepackage{graphics}
\usepackage{epsfig}
\usepackage{epstopdf}
\usepackage{amsmath}
\usepackage{latexsym}
\usepackage{wrapfig}
\usepackage{multirow}
\usepackage{amsfonts}
\usepackage{cite}
\usepackage{amssymb}
\usepackage{caption}
\usepackage{algorithm}
\usepackage{algorithmic}
\usepackage{enumerate}
\usepackage{diagbox}
\usepackage{subfig}
\usepackage{wrapfig}
\usepackage{color}
\usepackage{hyperref}

\usepackage[a4paper,margin=2.5cm,footskip=0.25in]{geometry}

\makeatletter
\newcommand{\rmnum}[1]{\romannumeral #1}
\newcommand{\Rmnum}[1]{\expandafter\@slowromancap\romannumeral #1@}
\makeatother

\spnewtheorem{claim}{Claim}{\bfseries}{\rmfamily}

\spnewtheorem{remark}{Remark}{\bfseries}{\rmfamily}

\begin{document}

\title{Minimum Enclosing Ball Revisited: Stability and Sub-linear Time Algorithms}
\author{Hu Ding}
\institute{
 School of Computer Science and Engineering, University of Science and Technology of China \\
 He Fei, China\\
  \email{huding@ustc.edu.cn}\\
}

\maketitle

\thispagestyle{empty}


\begin{abstract}
In this paper, we revisit the Minimum Enclosing Ball (MEB) problem and its robust version, MEB with outliers, in Euclidean space $\mathbb{R}^d$. Though the problem has been extensively studied before, most of the existing algorithms need at least linear time (in the number of input points $n$ and the dimensionality $d$) to achieve a $(1+\epsilon)$-approximation. Motivated by some recent developments on beyond worst-case analysis, we introduce the notion of stability for MEB (with outliers), which is natural and easy to understand. Roughly speaking, an instance of MEB  is stable, if the radius of the resulting ball cannot be significantly reduced by removing a small fraction of the input points. Under the stability assumption, we present two sampling algorithms for computing approximate MEB with sample complexities independent of the number of input points $n$. In particular, the second algorithm has the sample complexity even independent of the dimensionality $d$. Further, we extend the idea to achieve a sub-linear time approximation algorithm for the MEB with outliers problem. Note that most existing sub-linear time algorithms for the problems of MEB and MEB with outliers usually result in bi-criteria approximations, where the ``bi-criteria'' means that the solution has to allow the approximations on the radius and the number of covered points. Differently, all the algorithms proposed in this paper yield single-criterion approximations (with respect to radius). 
We expect that our proposed notion of stability and techniques will be applicable to design sub-linear time algorithms for other optimization problems. 
\end{abstract}

%
%
%
  
  \newpage

\pagestyle{plain}
\pagenumbering{arabic}
\setcounter{page}{1}

\section{Introduction}
\label{sec-intro}

Given a set $P$ of $n$ points in Euclidean space $\mathbb{R}^d$, where $d$ could be quite high, the problem of {\em Minimum Enclosing Ball (MEB)} is to find 
a ball with minimum radius to cover all the points in $P$ \cite{badoiu2003smaller,DBLP:journals/jea/KumarMY03,DBLP:conf/esa/FischerGK03}. MEB is a fundamental problem in computational geometry and finds 
applications in many fields such as machine learning and data mining. For example, one of the most popular classification models, {\em Support  Vector Machine (SVM)}, can be formulated as an MEB problem in high dimensional space, and 
fast MEB algorithms can be adopted to speed up 
its training procedure \cite{DBLP:journals/jmlr/TsangKC05,C10,DBLP:journals/jacm/ClarksonHW12}.
Recently, 
MEB has also been used for preserving privacy~\cite{DBLP:conf/pods/NissimSV16,DBLP:conf/ipsn/FeldmanXZR17} and quantum cryptography~\cite{gyongyosi2014geometrical}.

In real-world applications,  we often need to assume the presence of outliers in 
given datasets. MEB with outliers is a natural generalization of the MEB problem, where the goal is to find the minimum ball covering at least a certain  fraction or number of input points; for example, the ball may be required to cover at least $90\%$ of the points and  leave the remaining $10\%$ of points as outliers. The presence of outliers makes the problem  not only non-convex but also highly combinatorial; the high dimensionality of the problem further increases its challenge. 

The problems of MEB and MEB with outliers have been extensively studied before. However, almost all of them need at least linear time (in terms of $n$ and $d$) to obtain a $(1+\epsilon)$-approximation. This is not quite ideal, especially in big data where the size of the dataset could be so large that we cannot even afford to read the whole dataset once. This motivates us to ask 
the following question: is it possible to develop approximation algorithms for MEB (with outliers) that run in sub-linear time in the input size? 
Designing sub-linear time algorithms has become a promising approach to handle many big data problems and has attracted a great deal of attentions in the past decades \cite{rubinfeld2006sublinear,czumaj2006sublinear}. Note that most existing sampling based sub-linear time algorithms for MEB and MEB with outliers often result in errors on the number of covered points; for instance, to guarantee the desired quality of the ball, it may need to cover less than $n$ points (or discard more than the pre-specified number of outliers for MEB with outliers). A detailed discussion on previous works is given in Section~\ref{sec-related}.



\subsection{Our Main Ideas and Results}
\label{sec-overview}

Our idea for designing sub-linear time MEB (with outliers) algorithms is inspired by some recent developments on optimization with respect to stable instances, under the umbrella of {\em beyond worst-case analysis} \cite{DBLP:journals/cacm/Roughgarden19}. 
Many NP-hard optimization problems have shown to be challenging even for approximation, but 
  admit efficient solutions 
in practice. Several recent works tried to explain this phenomenon and introduced the notion of stability for problems like clustering and max-cut~\cite{balcan2013clustering,DBLP:journals/cpc/BiluL12,ostrovsky2012effectiveness,DBLP:conf/focs/AwasthiBS10}. 
In this paper, we give the notion of ``stability'' for MEB. Roughly speaking, an instance of MEB  is stable, if the radius of the resulting ball cannot be significantly reduced by removing a small fraction of the input points ({\em e.g.,} the radius cannot be reduced by $10\%$ if only $1\%$ of the points are removed). The rationale behind 
this notion is quite natural: if the given instance is not stable, the small fraction of points causing  significant reduction in the radius should be viewed as outliers (or they may need to be covered by 
additional balls, {\em e.g.,} $k$-center clustering \cite{hochbaum1985best,gonzalez1985clustering}). To the best of our knowledge, this is the first study on MEB (with outliers) from the perspective of stability.

We prove an important implication of the stability assumption: informally speaking, if an instance of MEB is stable, its center should reveal a certain extent of robustness in the space (Section~\ref{sec-stable}). The implication is useful  not only for designing sub-linear time MEB (with outliers) algorithms, but also for handling incomplete datasets. Using 
this implication,
we propose two sampling algorithms for computing approximate MEB with sub-linear time complexities (Section~\ref{sec-sub}); in particular, our second algorithm has the sample size ({\em i.e.,} the number of sampled points) independent of the input size $n$ and dimensionality $d$. 
We further consider the MEB with outliers problem of stable instances, and obtain a sub-linear time  approximation algorithm in Section~\ref{sec-outlier-stable2}. Besides the improvement on the sample size, another key difference with most existing sub-linear time MEB (with outliers) algorithms is that our algorithms yield \textbf{single-criterion} approximations that have no error on the number of covered points (see our discussion on existing ``bi-criteria approximation'' algorithms in Section~\ref{sec-related}). 

\textbf{Note} that if we arbitrarily select a point from the input dataset, it will be the center of a $2$-approximate MEB by the triangle inequality. But we are more interested in computing a  solution with approximation ratio close to $1$. Moreover, it is challenging to determine the radius of the ball in sub-linear time. In some applications, only estimating the position of the ball center may not be sufficient, and a ball covering all the given points is thus needed ({\em e.g.,} Tsang {\em et al.}~\cite{DBLP:conf/icml/TsangKK07} used enclosing ball to solve the SVM problem). \textbf{In this paper, we aim to determine not only the center of the ball, but also its radius, in sub-linear time}.

\subsection{Related Work}
\label{sec-related}

The works most related to ours  are \cite{alon2003testing,DBLP:journals/jacm/ClarksonHW12}. Alon {\em et al.}~\cite{alon2003testing} studied the following property testing problem: given a set of $n$ points in some metric space, 
determine 
whether the instance is $(k, b)$-clusterable, where an instance is called $(k, b)$-clusterable if it can be covered by $k$ balls with radius (or diameter) $b>0$. They proposed several sampling algorithms to answer the question ``approximately''. Particularly, they distinguish between the case that the instance is $(k, b)$-clusterable and the case that it is $\epsilon$-far away from $(k, b')$-clusterable, where $\epsilon\in (0,1)$ and $b'\geq b$. ``$\epsilon$-far'' means that more than $\epsilon n$ points should be removed so that it becomes $(k, b')$-clusterable. 
Note that their method cannot yield an approximation algorithm for the MEB problem, since it will introduce an unavoidable error on the number of covered points due to the relaxation of ``$\epsilon$-far''. However, it is possible to convert it into a ``\textbf{bi-criteria}'' approximation algorithm for the MEB with outliers problem, where it allows approximations on the radius and the number of uncovered outliers ({\em e.g.,} the solution may discard slightly more than the pre-specified number of outliers). 


Clarkson {\em et al.}~\cite{DBLP:journals/jacm/ClarksonHW12} developed an elegant perceptron framework for solving several optimization problems arising in machine learning, such as MEB.  For a set of $n$ points in $\mathbb{R}^d$ represented as  
an $n\times d$ matrix with $M$ non-zero entries, 
their framework can solve the MEB problem in $\tilde{O}(\frac{n}{\epsilon^2}+\frac{d}{\epsilon})$~\footnote{The asymptotic notation $\tilde{O}(f)=O\big(f\cdot polylog(\frac{nd}{\epsilon})\big)$.} time. Note that the parameter ``$\epsilon$'' is an additive error ({\em i.e.,} the resulting radius is $r+\epsilon$ if $r$ is the radius of the optimal MEB) which can be converted into a relative error ({\em i.e.,} $(1+\epsilon)r$) in $O(M)$ preprocessing time. Thus, if $M=o(nd)$, the running time is still sub-linear in the input size $nd$. Our algorithms have different sub-linear time complexities which are independent of the number of input points. 

%

\textbf{MEB and core-set.} A {\em core-set}~\cite{agarwal2005geometric} is a small set of points that approximates the structure/shape of a much larger point set, and thus can be used to significantly reduce the time complexities for many optimization problems (the reader is referred to the recent surveys~\cite{DBLP:journals/corr/Phillips16,DBLP:journals/widm/Feldman20} for more details on core-sets). The core-set idea  has also been used to approximate the MEB problem in high dimensional space~\cite{BHI,DBLP:journals/jea/KumarMY03,panigrahy2004minimum,DBLP:conf/isaac/KerberS13}. B\u{a}doiu and Clarkson \cite{badoiu2003smaller} showed that it is possible to find a core-set of size $\lceil2/\epsilon\rceil$ that yields a $(1+\epsilon)$-approximate MEB
; later,  they \cite{coresets1} further proved that actually only $\lceil 1/\epsilon\rceil$ points are sufficient, but their core-set construction is more complicated. In fact, the algorithm for computing the core-set of MEB is a {\em Frank-Wolfe} style algorithm~\cite{frank1956algorithm}, which has been systematically studied by Clarkson~\cite{C10}. 
There are also several exact and approximation algorithms for MEB that do not rely on core-sets
~\cite{DBLP:conf/esa/FischerGK03,DBLP:conf/soda/SahaVZ11,DBLP:conf/icalp/ZhuLY16}. Most of these algorithms have linear time complexities. Agarwal and Sharathkumar~\cite{DBLP:journals/algorithmica/AgarwalS15} presented a streaming $(\frac{1+\sqrt{3}}{2}+\epsilon)$-approximation algorithm for computing MEB; later, Chan and Pathak~\cite{DBLP:journals/comgeo/ChanP14} proved that the same algorithm actually yields an approximation ratio less than $1.22$.

\textbf{MEB with outliers and bi-criteria approximations.} 
The MEB with outliers problem is a generalization of the MEB problem that allows to discard a certain fraction of points as outliers; it also can be viewed as the case $k=1$ of the $k$-center clustering with outliers problem~\cite{charikar2001algorithms}. 
B\u{a}doiu {\em et al.}~\cite{BHI} extended their core-set idea to the problems of MEB and $k$-center clustering with outliers, and achieved linear time bi-criteria approximation algorithms (if $k$ is assumed to be a constant). Recently, Huang {\em et al.}~\cite{DBLP:conf/focs/HuangJLW18} and  Ding {\em et al.}~\cite{DBLP:journals/corr/abs-1901-08219} respectively showed that the simple uniform sampling approach yields bi-criteria approximation of $k$-center clustering with outliers, but their sample sizes both depend on the dimensionality $d$. In~\cite{DBLP:journals/corr/abs-2004-10090}, we  proposed a unified framework for designing sub-linear time bi-criteria approximation algorithms for several high-dimensional geometric optimization problems with outliers, where the sample sizes can be improved to be independent of $d$. 
Several algorithms for the low dimensional MEB with outliers problem have also been developed~\cite{aggarwal1991finding,efrat1994computing,har2005fast,matouvsek1995enclosing}. 
 Furthermore, there exist a number of works on streaming MEB and $k$-center clustering with outliers, such as~\cite{charikar2003better,mccutchen2008streaming,zarrabistreaming}.

\textbf{Optimizations under stability.} Bilu and Linial~\cite{DBLP:journals/cpc/BiluL12} showed that the Max-Cut problem becomes easier if the given instance is stable with respect to perturbation on edge weights. Ostrovsky {\em et al.}~\cite{ostrovsky2012effectiveness} proposed a separation condition for $k$-means clustering which refers to the scenario where  the clustering cost of $k$-means  is significantly lower than that of $(k-1)$-means  for a given instance, and demonstrated the effectiveness of the Lloyd heuristic \cite{lloyd1982least} under the separation condition. Balcan {\em et al.}~\cite{balcan2013clustering} introduced the concept of approximation-stability for finding the ground-truth of $k$-median and $k$-means clustering. 
Awasthi {\em et al.}~\cite{DBLP:conf/focs/AwasthiBS10} introduced another notion of clustering stability and gave a PTAS for $k$-median and $k$-means clustering. More algorithms on clustering problems under stability assumption were studied in~\cite{DBLP:journals/ipl/AwasthiBS12,DBLP:journals/siamcomp/BalcanL16,DBLP:conf/icalp/BalcanHW16,DBLP:conf/colt/BalcanB09,kumar2010clustering}.

\textbf{Sub-linear time algorithms.} Indyk presented sub-linear time algorithms for several metric space problems, such as $k$-median clustering~\cite{DBLP:conf/stoc/Indyk99} and $2$-clustering~\cite{DBLP:conf/focs/Indyk99}. More sub-linear time clustering algorithms have been studied in~\cite{meyerson2004k,mishra2001sublinear,czumaj2004sublinear}. Another important motivation for designing sub-linear time algorithms is property testing. For example, Goldreich {\em et al.}~\cite{DBLP:journals/jacm/GoldreichGR98} focused on using small sample to test some natural graph properties. More detailed discussion on sub-linear time algorithms can be found in the survey papers~\cite{rubinfeld2006sublinear,czumaj2006sublinear}.

\vspace{-0.08in}

\section{Definitions and Preliminaries}
\label{sec-pre}
\vspace{-0.08in}

In this paper, we let $|A|$ denote the number of points of a given point set $A$ in $\mathbb{R}^d$, and $||x-y||$ denote the Euclidean distance between two points $x$ and $y$ in $\mathbb{R}^d$. We use $\mathbb{B}(c, r)$ to denote the ball centered at a point $c$ with radius $r>0$. Below, we first give the definitions of MEB and the notion of stability. 

\begin{definition}[Minimum Enclosing Ball (MEB)]
\label{def-meb}
Given a set $P$ of $n$ points in $\mathbb{R}^d$, the MEB problem is to find a ball with minimum radius to cover all the points in $P$. The resulting ball and its radius are denoted by $MEB(P)$ and $Rad(P)$, respectively.
\end{definition}

A ball $\mathbb{B}(c, r)$ is called a $\lambda$-approximation of $MEB(P)$ for some $\lambda\geq 1$, if the ball covers all points in $P$ and has radius $r\leq \lambda Rad(P)$. Unless otherwise specified, we always assume that $\epsilon$ is a fixed number in $(0,1)$ that will be often used to measure the error or change on radius in the rest of the paper.

\begin{definition}[$\beta_\epsilon$-stable]
\label{def-stable}
Given a set $P$ of $n$ points in $\mathbb{R}^d$ with a small parameter $\beta_\epsilon\in (0,1)$,  $P$ is a $\beta_\epsilon$-stable instance if $Rad(P')\geq (1-\epsilon)Rad(P)$ for any $P'\subset P$ with $|P'|\geq (1-\beta_\epsilon)n$. 
\end{definition}


\textbf{The intuition of Definition~\ref{def-stable}.} The property of stability indicates that $Rad(P)$ cannot be significantly reduced unless removing a large enough fraction of points from $P$. In particular, the larger $\beta_\epsilon$ is, the more stable $P$ becomes.   Actually, our stability assumption is quite reasonable  in practice. For example, if the radius of MEB can be reduced considerably (say by $\epsilon=10\%$) after removing only a very small fraction (say $\beta_\epsilon=1\%$) of  points, 
it is natural to view the small fraction of points as outliers. 
To better understand the notion of stability in high dimensions, we consider the following two examples.

\textbf{ Example \rmnum{1}.} Suppose that the distribution of $P$ is uniform and dense inside $MEB(P)$. If we want the radius of the remaining $1-\beta_\epsilon$ points to be as small as possible, intuitively we should remove the outermost $\beta_\epsilon$ points (since $P$ is uniform and dense).  Let $P'$  denote  the set of innermost $1-\beta_\epsilon$ points. Assume $Rad(P')=(1-\epsilon)Rad(P)$. Then we have $\frac{|P'|}{|P|}\approx \frac{Vol\big(MEB(P')\big)}{Vol\big(MEB(P)\big)}=\frac{(Rad(P'))^d}{(Rad(P))^d}= (1-\epsilon)^d$, where $Vol(\cdot)$ is the volume function. That is, $1-\beta_\epsilon\approx(1-\epsilon)^d$ and thus $\lim_{d\to\infty}\beta_\epsilon=1$; that means $P$ tends to be very stable as $d$ increases in this case. 

\textbf{ Example \rmnum{2}.} Consider a regular $d$-dimensional  simplex $P$ containing $d+1$ points where each pair of points have the pairwise distance equal to $1$. It is known that $Rad(P)$ is $r_d=\sqrt{\frac{d}{2(1+d)}}$. If we remove $\beta_\epsilon (d+1)$ points from $P$, namely it becomes a regular $d'$-dimensional simplex with $d'=(1-\beta_\epsilon)(d+1)-1$, the new radius $r_{d'}=\sqrt{\frac{d'}{2(1+d')}}$. To obtain $\frac{r_{d'}}{r_d}\leq 1-\epsilon$, it is easy to see that $1-\beta_\epsilon\leq\frac{1}{1+(2\epsilon-\epsilon^2) d}$ and thus $\lim_{d\to\infty}\beta_\epsilon=1$. Similar to example \rmnum{1}, the instance  
$P$ tends to be very stable as $d$ increases. 

\begin{remark}
In practice, it may be challenging to obtain the exact value of $\beta_\epsilon$ for a given instance. However, the value of $\beta_\epsilon$ only affects the sample size in our following proposed algorithms ({\em e.g.,} the algorithms in Section~\ref{sec-sub}), and it is sufficient to only estimate a lower bound instead of the exact value. Moreover, it is possible to estimate a reasonable range of $\beta_\epsilon$ by using some prior knowledge on the input data or the techniques proposed in~\cite{alon2003testing}.
\end{remark}

%
%

%


\begin{definition}[MEB with Outliers]
\label{def-outlier}
Given a set $P$ of $n$ points in $\mathbb{R}^d$ and a small parameter $\gamma\in (0,1)$, the MEB with outliers problem is to find the smallest ball that covers $(1-\gamma)n$ points. Namely, the task is to find a subset of $P$ with size $(1-\gamma)n$ such that the resulting MEB is the smallest among all possible choices of the subset. The obtained ball is denoted by $MEB(P, \gamma)$.
\end{definition}

For convenience, we  use $P_{\textnormal{opt}}$ to denote the optimal subset of $P$ with respect to $MEB(P, \gamma)$. That is,
$P_{\textnormal{opt}}=\arg_{Q}\min\Big\{Rad(Q)\mid Q\subset P, \left|Q\right|= (1-\gamma)n\Big\}$. 
From Definition~\ref{def-outlier}, we can see that the main issue is to determine the subset of $P$. We also extend Definition~\ref{def-stable} of MEB to MEB with outliers. 

%
%
%
%
%
%
%

%

\begin{definition}[$\beta_\epsilon$-stable for MEB with Outliers]
\label{def-outlier-stable}
Let $\beta_\epsilon\in (0,1)$. 
Given an instance $(P, \gamma)$ of the  MEB with outliers problem in Definition~\ref{def-outlier}, $(P, \gamma)$ is a $\beta_\epsilon$-stable instance if $Rad(P')\geq (1-\epsilon)Rad(P_{opt})$ for any $P'\subset P$ with $|P'|\geq \big(1-\gamma-\beta_\epsilon\big)n$.
\end{definition}

Definition~\ref{def-outlier-stable} directly implies the following claim.

\begin{claim}
\label{cla-outlier-stable}
If $(P, \gamma)$ is a $\beta_\epsilon$-stable instance of the problem of MEB with outliers, the corresponding $P_{opt}$ is a $\frac{\beta_\epsilon}{1-\gamma}$-stable instance of MEB.
\end{claim}

To see the correctness of 
Claim~\ref{cla-outlier-stable}, we can use contradiction. 
Suppose that there exists a subset $P'\subset P_{opt}$ such that $|P'|\geq (1-\frac{\beta_\epsilon}{1-\gamma})|P_{opt}|=(1-\gamma-\beta_\epsilon)n$ and $Rad(P')<(1-\epsilon)Rad(P_{opt})$. Then, it is in contradiction to the fact that $(P, \gamma)$ is a $\beta_\epsilon$-stable instance of MEB with outliers. 


\vspace{-0.05in}
\subsection{Core-set Construction for MEB~\cite{badoiu2003smaller}}
\label{sec-newanalysis}
\vspace{-0.05in}
To compute a $(1+\epsilon)$-approximate MEB, B\u{a}doiu and Clarkson proposed an algorithm yielding an MEB core-set of size $2/\epsilon$ (for convenience, we always assume that $2/\epsilon$ is an integer)~\cite{badoiu2003smaller}. 
We first briefly introduce their main idea, since it will be  used in our proposed algorithms. 
 



Given a point set $P\subset\mathbb{R}^d$, the algorithm is a simple iterative procedure. Initially, it selects an arbitrary point from $P$ and places it into an initially empty set $T$. 
In each of the following $2/\epsilon$ iterations, the algorithm updates the center of $MEB(T)$ and adds to $T$ the farthest point from the current center of $MEB(T)$. 
Finally, the center of $MEB(T)$ induces a $(1+\epsilon)$-approximation for $MEB(P)$. The selected set of $2/\epsilon$ points ({\em i.e.}, $T$) is called the core-set of MEB. However, computing the exact center of $MEB(T)$ could be expensive; in practice, one may only compute an approximate center of $MEB(T)$ in each iteration. 
In the $i$-th iteration, we let $c_i$ denote the exact center of $MEB(T)$; also, let $r_i$ be the radius of $MEB(T)$. Suppose $\xi$ is a given number in $(0,1)$. Using another algorithm proposed in~\cite{badoiu2003smaller}, one can compute an approximate center $o_i$ having the distance to $c_i$ less than $\xi r_i$ in $O(\frac{1}{\xi^2}|T|d)$ time. Since we only compute $o_i$ rather than $c_i$ in each iteration, we in fact only select the farthest point to $o_i$ (not $c_i$). 
%
In~\cite{DBLP:journals/corr/abs-2004-10090}, we studied this modification on B\u{a}doiu and Clarkson's method and presented the following theorem. For the sake of completeness, we place the proof from~\cite{DBLP:journals/corr/abs-2004-10090} in Section~\ref{sec-proof-newbc}. 


\begin{theorem}[\cite{DBLP:journals/corr/abs-2004-10090}]
\label{the-newbc}
In the core-set construction algorithm of~\cite{badoiu2003smaller}, if one computes an approximate MEB for $T$ in each iteration and the resulting center $o_i$ has the distance to $c_i$ less than $\xi r_i= s\frac{\epsilon}{1+\epsilon} r_i$ for some $s\in(0,1)$, the final core-set size is bounded by $z=\frac{2}{(1-s)\epsilon}$. Also, the bound could be arbitrarily close to $2/\epsilon$ when $s$ is small enough. 
\end{theorem}

\begin{remark}
\label{rem-newbc}
\textbf{(\rmnum{1})} We can simply set $s$ to be any constant in $(0,1)$; for instance, if $s=1/3$, the core-set size will be bounded by $z=3/\epsilon$. Since $|T|\leq z$, the total running time is $O\Big(z\big(|P|d+\frac{1}{\xi^2}zd\big)\Big)=O\Big(\frac{1}{\epsilon}\big(|P|+\frac{1}{\epsilon^3}\big)d\Big)$. 
\textbf{(\rmnum{2})} We also want to emphasize a simple observation mentioned in~\cite{DBLP:journals/corr/abs-2004-10090} on the above core-set construction procedure, which will be used in our algorithms and analysis later on. The algorithm always selects the farthest point to $o_i$ in each iteration. However, this is actually not necessary. As long as the selected point has distance at least $(1+\epsilon)Rad(P)$, the result presented in Theorem~\ref{the-newbc} is still true.
 If no such a point exists ({\em i.e.,} $P\setminus \mathbb{B}\big(o_i, (1+\epsilon)Rad(P)\big)=\emptyset$), a $(1+\epsilon)$-approximate MEB ({\em i.e.,} $\mathbb{B}\big(o_i, (1+\epsilon)Rad(P)\big)$) has already been obtained.  
\end{remark}

\section{Implication of the Stability Property}
\label{sec-stable}

\begin{wrapfigure}{r}{0.28\textwidth}
 \vspace{-0.45in}
\begin{center}
    \includegraphics[width=0.15\textwidth]{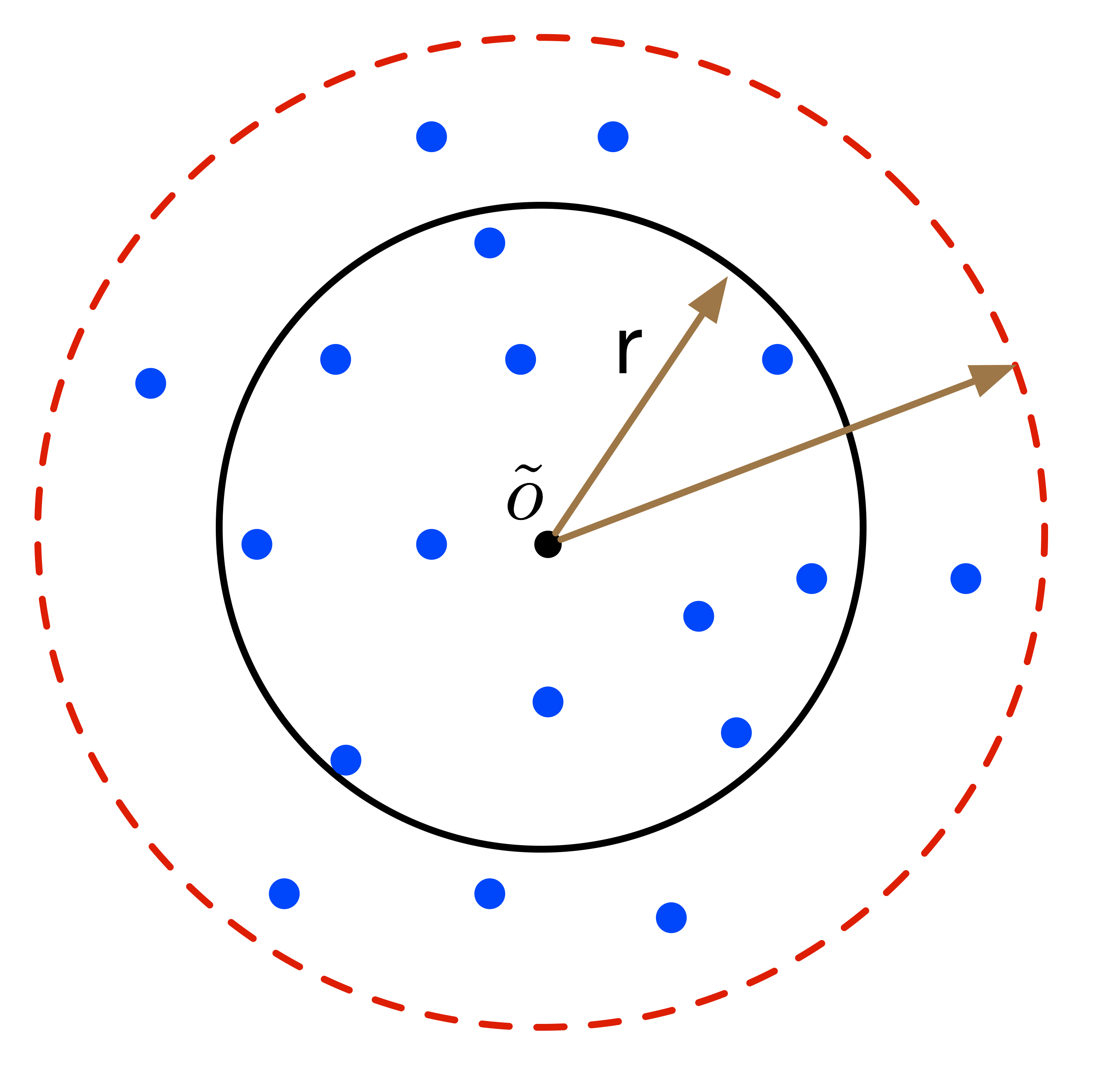}  
    \end{center}
  \vspace{-12pt}
  \caption{We expand $\mathbb{B}(\tilde{o}, r)$, and the larger ball with radius $\frac{1+(2+\sqrt{2})\sqrt{\epsilon}}{1-\epsilon} r$ is an approximate MEB of $P$.}     
   \label{fig-im}
     \vspace{-25pt}
\end{wrapfigure}

In this section, we show an important implication of the stability property described in Definition~\ref{def-stable}.

\begin{theorem}
\label{the-stable}
Let $P$ be a  $\beta_\epsilon$-stable instance of the MEB problem, and $o$ be the center of its MEB. 
Let $\tilde{o}$ be a given point in $\mathbb{R}^d$. Assume the number $r\leq (1+\epsilon)Rad(P)$. If the ball $\mathbb{B}\big(\tilde{o}, r\big)$ covers at least $(1-\beta_\epsilon)n$ points from $P$, the following holds
\begin{eqnarray}
||\tilde{o}-o||&<&  (2+\sqrt{2})\sqrt{\epsilon} Rad(P)=O(\sqrt{\epsilon}) Rad(P).\label{for-stable}
\end{eqnarray}
\end{theorem}

%

Theorem~\ref{the-stable} indicates that if a ball covers a large enough subset of $P$ and its radius is bounded, its center should be close to the center of $MEB(P)$. 
Actually, besides using it to design our sub-linear time MEB algorithms later, Theorem~\ref{the-stable} is also useful in other practical scenarios. Suppose we have $\beta_\epsilon n$ (or less) missing points from $P$. We compute a $(1+\epsilon)$-approximate MEB of the remaining $(1-\beta_\epsilon) n$ points and use $\mathbb{B}(\tilde{o}, r)$ to denote the obtained ball. Since the ball is a $(1+\epsilon)$-approximate MEB of a subset of $P$, we have $r\leq (1+\epsilon)Rad(P)$. Moreover, due to Definition~\ref{def-stable}, we know $r\geq (1-\epsilon)Rad(P)$. Together with Theorem~\ref{the-stable}, we have
\begin{eqnarray}
P\underbrace{\subset}_{\text{by }(\ref{for-stable})} \mathbb{B}\Big(\tilde{o}, \big(1+(2+\sqrt{2})\sqrt{\epsilon}\big) Rad(P)\Big)\underbrace{\subset}_{\text{by } r\geq (1-\epsilon)Rad(P)} \mathbb{B}\Big(\tilde{o}, \frac{1+(2+\sqrt{2})\sqrt{\epsilon}}{1-\epsilon} r\Big)
\end{eqnarray}
and the radius $\frac{1+(2+\sqrt{2})\sqrt{\epsilon}}{1-\epsilon} r\leq \frac{1+(2+\sqrt{2})\sqrt{\epsilon}}{1-\epsilon}(1+\epsilon)Rad(P)=\frac{1+O(\sqrt{\epsilon})}{1-\epsilon}Rad(P)$. 
That is, the ball $ \mathbb{B}\Big(\tilde{o}, \frac{1+(2+\sqrt{2})\sqrt{\epsilon}}{1-\epsilon} r\Big)$ is a $\frac{1+O(\sqrt{\epsilon})}{1-\epsilon}$-approximate MEB of $P$ (see Figure~\ref{fig-im}). Note that we cannot directly use $\mathbb{B}\Big(\tilde{o}, \big(1+(2+\sqrt{2})\sqrt{\epsilon}\big) Rad(P)\Big)$ since we do not know the value of $Rad(P)$. That is, we have to determine the radius based on the known values $r$ and $\epsilon$. Based on the above analysis, even if we have $\beta_\epsilon n$ missing points, we are still  able to compute an approximate MEB of $P$ with the approximation ratio $\frac{1+O(\sqrt{\epsilon})}{1-\epsilon}=1+O(\sqrt{\epsilon})$ (if $\epsilon$ is small enough). 


Now, we prove Theorem~\ref{the-stable}. Let $P'=\mathbb{B}\big(\tilde{o}, r\big)\cap P$, and assume $o'$ is the center of $MEB(P')$.
To bound the distance between $\tilde{o}$ and $o$, we need to bridge them by the point $o'$ (since $||\tilde{o}-o||\leq ||\tilde{o}-o'||+||o'-o||$).  
The following are two key lemmas to the proof.

\begin{lemma}
\label{lem-stable1}
The distance $||o'-o||\leq \sqrt{2\epsilon} Rad(P)$.
\end{lemma}
\begin{proof}

We consider two cases: $MEB(P')$ is totally covered by $MEB(P)$ and otherwise. 
For the first case (see Figure~\ref{fig-lem1-1}), it is easy to see that 
\begin{eqnarray}
||o'-o||\leq Rad(P)-(1-\epsilon) Rad(P)=\epsilon Rad(P)<\sqrt{2\epsilon} Rad(P), \label{for-lem-stable1-1}
\end{eqnarray}
where the first inequality comes from the fact that $MEB(P')$ has radius at least $(1-\epsilon) Rad(P)$ (Definition~\ref{def-stable}), and the last inequality comes from the fact that $\epsilon<1$. Thus, we can focus on the second case below.


\begin{figure}
  \vspace{-33pt}
  \begin{center}
  \subfloat[]{\label{fig-lem1-1}\includegraphics[width=0.13\textwidth]{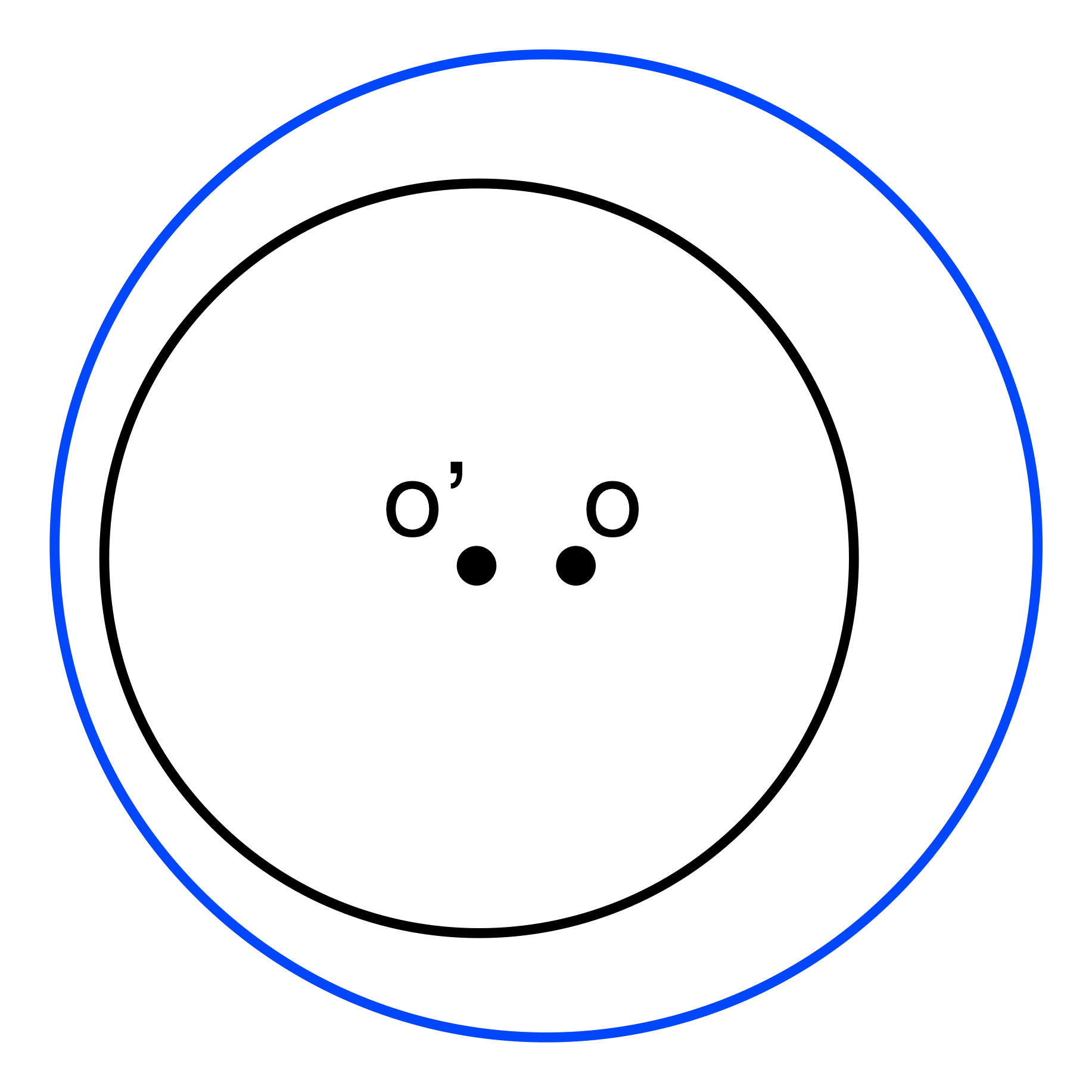}}
  \hspace{0.4in}
  \subfloat[]{\label{fig-lem1-2}\includegraphics[width=0.13\textwidth]{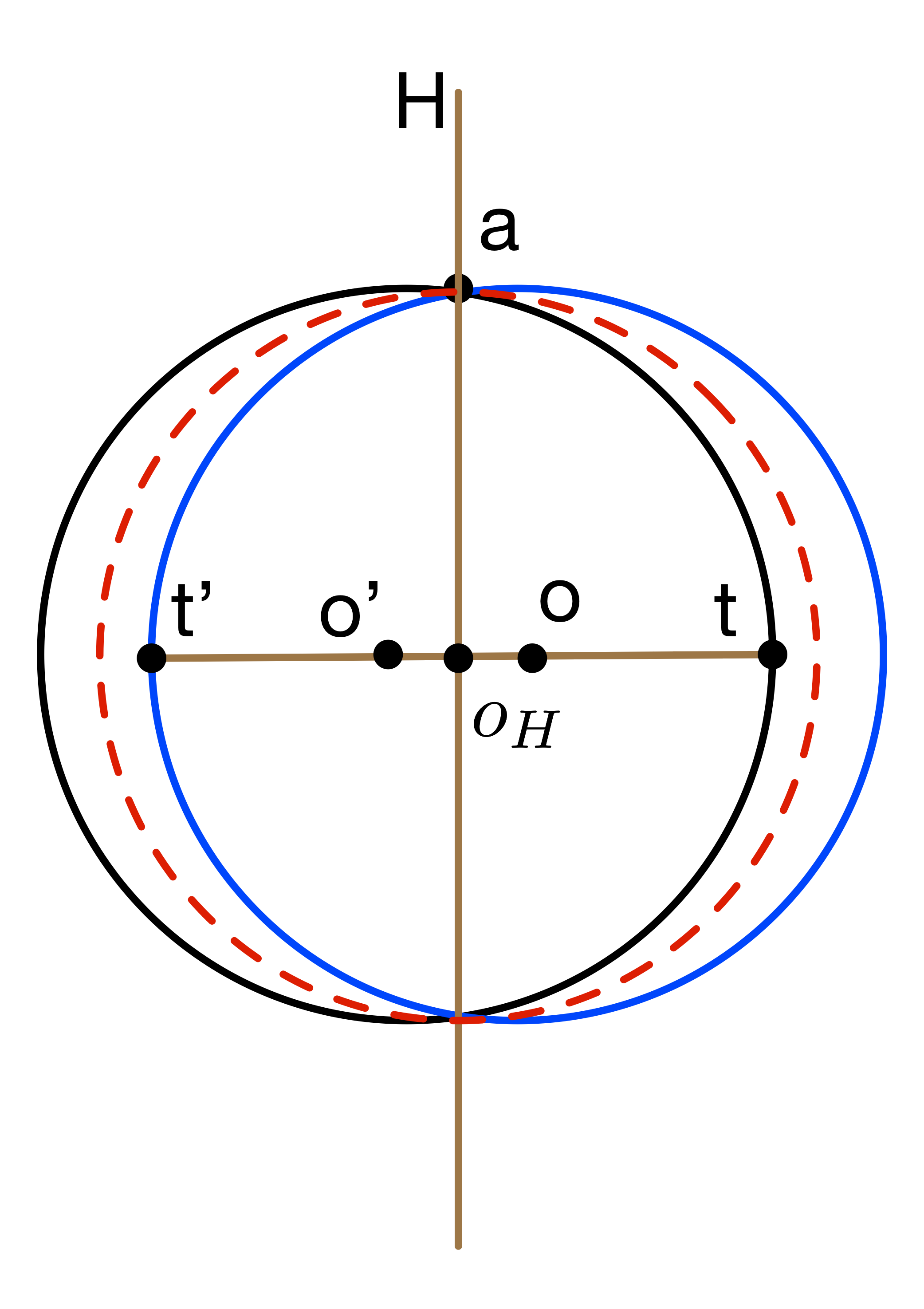}}
 \hspace{0.4in}
  \subfloat[]{\label{fig-lem1-3}\includegraphics[width=0.13\textwidth]{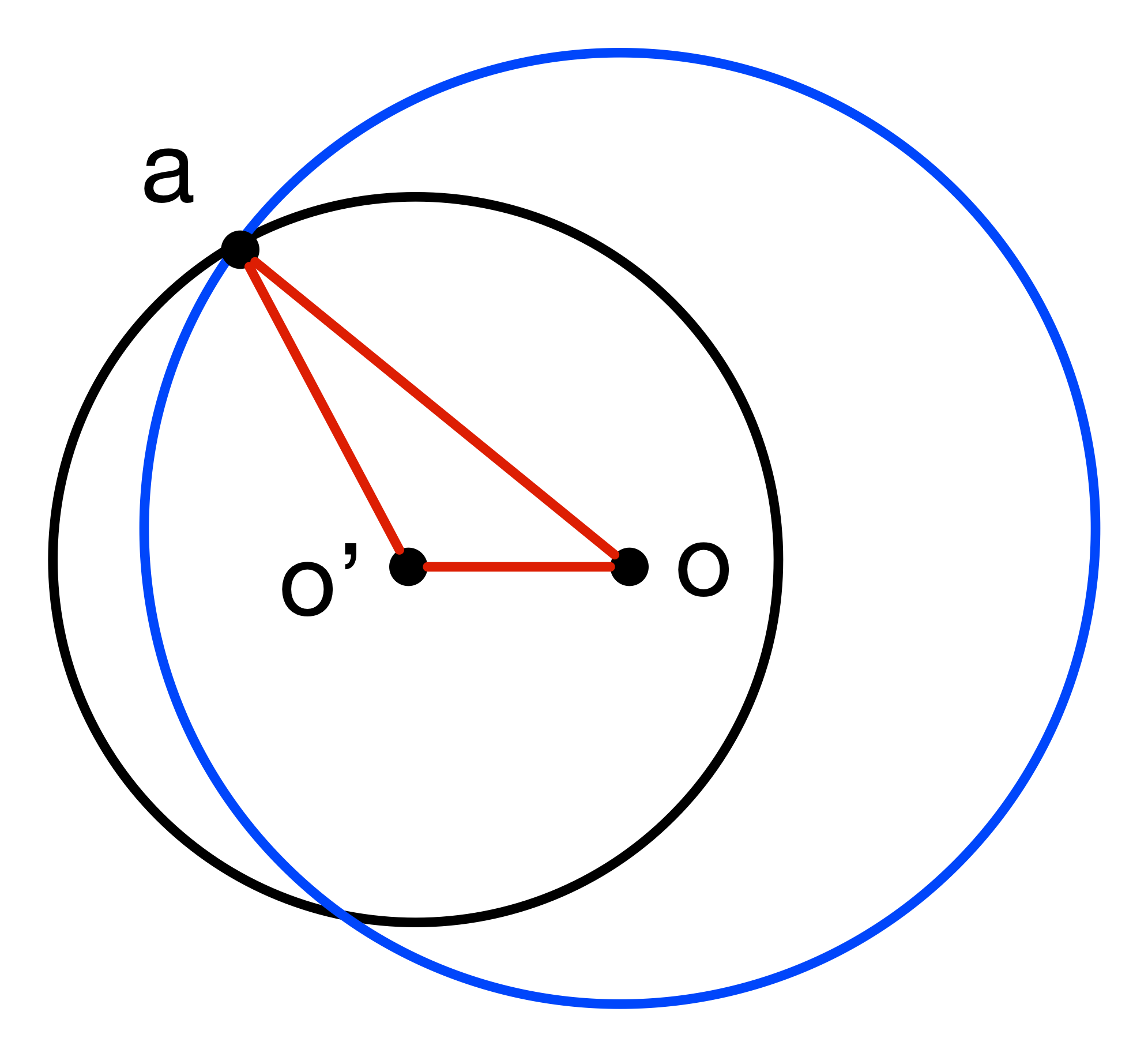}}
\hspace{0.4in}
  \subfloat[]{\label{fig-lem2}\includegraphics[width=0.13\textwidth]{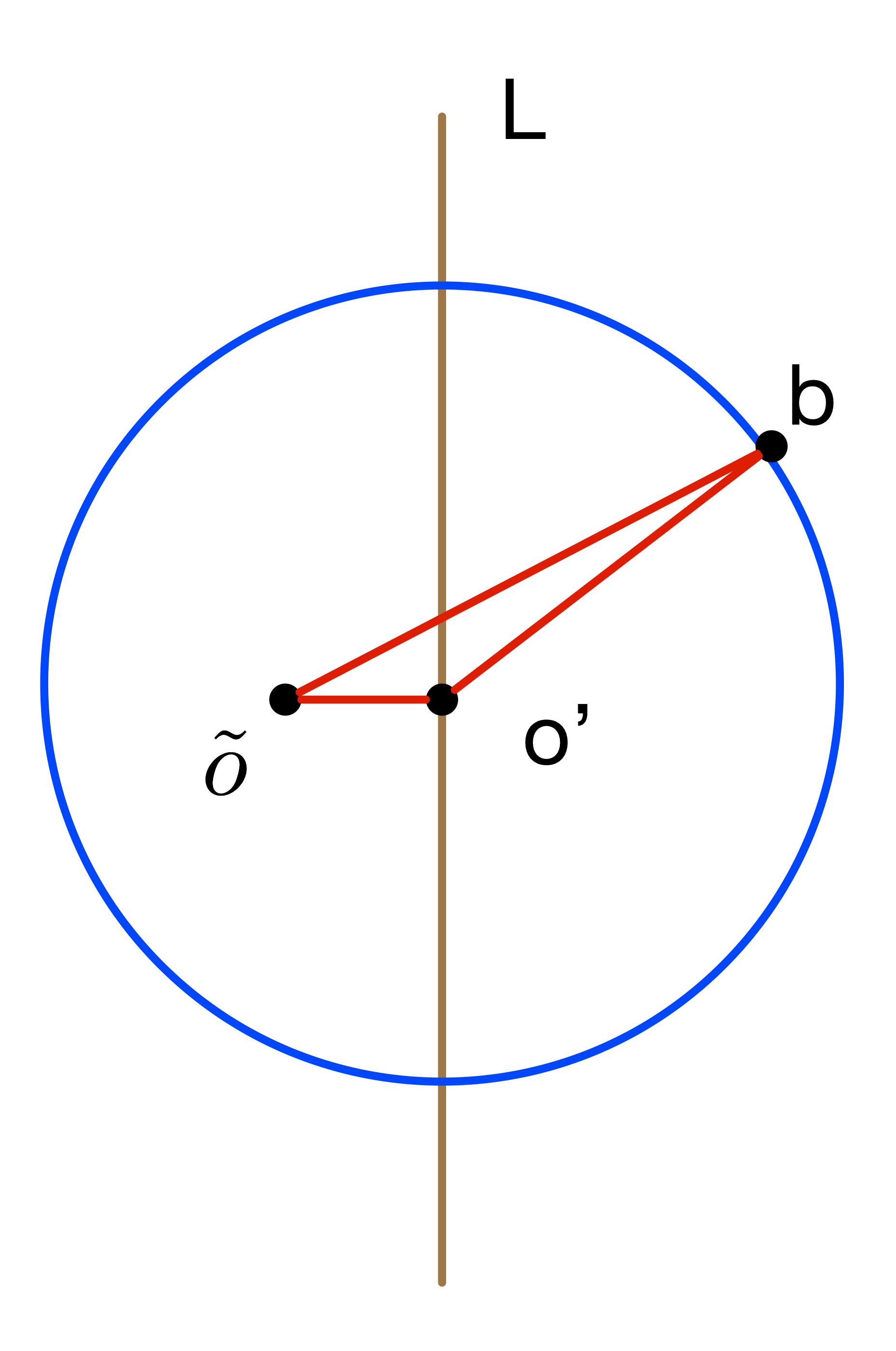}}
   \end{center}
 \vspace{-10pt}
  \caption{(a) The case $MEB(P')\subset MEB(P)$; (b) an illustration of Claim~\ref{cla-angle}; (c) the angle $\angle ao'o\geq \pi/2$; (d) an illustration of Lemma~\ref{lem-stable2}.}
     \vspace{-0.2in}
\end{figure}

Let $a$ be any point located on the intersection of the two spheres of $MEB(P')$ and $MEB(P)$. Consequently, we have the following claim. 

\begin{claim}
\label{cla-angle}
The angle $\angle ao'o\geq \pi/2$.
\end{claim}
\begin{proof}


Suppose  that $\angle ao'o< \pi/2$. Note that $\angle aoo'$ is always smaller than $\pi/2$ since $||o-a||=Rad(P)\geq Rad(P')=||o'-a||$. Therefore, $o$ and $o'$ are separated by the hyperplane $H$ that is  orthogonal 
to the segment $\overline{o'o}$ and passing through the point $a$. See Figure~\ref{fig-lem1-2}.

Now we  show that $P'$ can be covered by a ball smaller than $MEB(P')$. Let $o_H$ be the point $H\cap \overline{o'o}$, and $t$ ({\em resp.,} $t'$) be the point collinear with $o$ and $o'$ on the right side of the sphere of $MEB(P')$ ({\em resp.,} left side of the sphere of $MEB(P)$; see Figure ~\ref{fig-lem1-2}). Then, we have
\begin{eqnarray}
||t-o_H||+||o_H-o'||&=&||t-o'||=||a-o'||\nonumber\\
&<&||o'-o_H||+||o_H-a||\nonumber\\
\Longrightarrow ||t-o_H||&<&||o_H-a||.\label{cla-1}
\end{eqnarray}
Similarly, we have 
$||t'-o_H||<||o_H-a||$. 
Consequently, $MEB(P)\cap MEB(P')$ is covered by the ball $\mathbb{B}(o_H, ||o_H-a||)$. Further, because $P'$ is covered by  $MEB(P)\cap MEB(P')$ and $||o_H-a||<||o'-a||=Rad(P')$, $P'$ is covered by the ball $\mathbb{B}(o_H, ||o_H-a||)$ that is smaller than $MEB(P')$. This  contradicts to the fact that $MEB(P')$ is the minimum enclosing ball of $P'$. Thus, the claim $\angle ao'o\geq \pi/2$ is true.
\qed\end{proof}

Given Claim~\ref{cla-angle}, we know that $||o'-o||\leq \sqrt{\big(Rad(P)\big)^2-\big(Rad(P')\big)^2}$. See Figure~\ref{fig-lem1-3}. Moreover, Definition~\ref{def-stable} implies that $Rad(P')\geq (1-\epsilon) Rad(P)$. Therefore, we have 
\begin{eqnarray}
||o'-o||&\leq&\sqrt{\big(Rad(P)\big)^2-\big((1-\epsilon) Rad(P)\big)^2}\nonumber\\
&=&\sqrt{2\epsilon-\epsilon^2} Rad(P)\leq\sqrt{2\epsilon} Rad(P).
\end{eqnarray}
\qed\end{proof}

\begin{lemma}
\label{lem-stable2}
The distance $||\tilde{o}-o'||\leq 2\sqrt{\epsilon} Rad(P)$.
\end{lemma}
\begin{proof}



Let $L$ be the hyperplane orthogonal to
the segment $\overline{\tilde{o}o'}$ and passing through the center $o'$. Suppose $\tilde{o}$ is located on the left side of $L$. Then, there exists a point $b\in P'$ located on the right closed semi-sphere of $MEB(P')$ divided by $L$ (this result was proved in~\cite{goel2001reductions,BHI} and see Lemma 2.2 in~\cite{BHI}; for completeness, we also state the lemma in Section~\ref{sec-bhi}). See Figure \ref{fig-lem2}. That is, the angle $\angle bo'\tilde{o}\geq \pi/2$. As a consequence, we have 
\begin{eqnarray}
||\tilde{o}-o'||\leq\sqrt{||\tilde{o}-b||^2-||b-o'||^2}. \label{for-lem-stable2-1}
\end{eqnarray}
Moreover, since $||\tilde{o}-b||\leq r\leq (1+\epsilon)Rad(P)$ and $||b-o'||= Rad(P')\geq (1-\epsilon) Rad(P)$, (\ref{for-lem-stable2-1}) implies that $||\tilde{o}-o'||\leq \sqrt{(1+\epsilon)^2-(1-\epsilon)^2} Rad(P)=2\sqrt{\epsilon} Rad(P)$. 
\qed\end{proof}

By triangle inequality and Lemmas~\ref{lem-stable1} and \ref{lem-stable2}, we immediately have 
\begin{eqnarray}
||\tilde{o}-o||&\leq& ||\tilde{o}-o'||+||o'-o||\nonumber\\
&\leq&(2+\sqrt{2})\sqrt{\epsilon} Rad(P).
\end{eqnarray}
This completes the proof of Theorem~\ref{the-stable}.

\section{Sub-linear Time Algorithms for MEB}
\label{sec-sub}
Using Theorem~\ref{the-stable}, we present two different sub-linear time sampling algorithms for computing MEB. The first one is simpler, but has a sample size ({\em i.e.,} the number of sampled points) depending on the dimensionality $d$, while the second one has a sample size independent of both $n$ and $d$. Following most of existing articles on sub-linear time algorithms ({\em e.g.,} ~\cite{meyerson2004k,mishra2001sublinear,czumaj2004sublinear}), in each sampling step of our algorithms, we always take the sample \textbf{independently and uniformly at random}.

\subsection{The First Sampling Algorithm}
\label{sec-sample1}
Algorithm~\ref{alg-meb1} is based on the theory of VC dimension and $\epsilon$-net~\cite{vapnik2015uniform,haussler1987}. 
Roughly speaking, we compute an approximate MEB of a small random sample ({\em i.e.,} $\mathbb{B}(c, r)$), and expand the ball slightly; then we prove that this expanded ball is an approximate MEB of the whole data set.  As emphasized in Section~\ref{sec-overview}, our result is a single-criterion approximation. If simply applying the uniform sample idea without the stability assumption (as the ideas in~\cite{alon2003testing,DBLP:conf/focs/HuangJLW18,DBLP:journals/corr/abs-1901-08219}), it will result in a bi-criteria approximation where the ball has to cover less than $n$ points for achieving the desired bounded radius. 
Our key idea is to show that $\mathbb{B}(c, r)$ covers at least $(1-\beta_\epsilon)n$ points and therefore $c$ is close to the optimal center by Theorem~\ref{the-stable}. 


\renewcommand{\algorithmicrequire}{\textbf{Input:}}
\renewcommand{\algorithmicensure}{\textbf{Output:}}
\begin{algorithm}
   \caption{MEB Algorithm \Rmnum{1}}
   \label{alg-meb1}
\begin{algorithmic}[1]
\REQUIRE A $\beta_\epsilon$-stable instance $P$ of MEB problem in $\mathbb{R}^d$.
\STATE Sample a set $S$ of $\Theta(\frac{d}{\beta_\epsilon}\log \frac{d}{\beta_\epsilon})$ points from $P$ uniformly at random.
\STATE Apply any approximate MEB algorithm (such as the core-set based algorithm~\cite{badoiu2003smaller}) to compute a $(1+\epsilon)$-approximate MEB of $S$, and let the resulting ball be $\mathbb{B}(c, r)$.
\STATE Output the ball $\mathbb{B}\big(c,\frac{1+(2+\sqrt{2})\sqrt{\epsilon}}{1-\epsilon}r\big)$.
\end{algorithmic}
\end{algorithm}

\begin{theorem}
\label{the-meb1}
With constant probability, Algorithm~\ref{alg-meb1} returns a $\lambda$-approximate MEB of $P$, where
\begin{eqnarray}
\lambda=\frac{\big(1+(2+\sqrt{2})\sqrt{\epsilon}\big)(1+\epsilon)}{1-\epsilon} 
\end{eqnarray}
and $\lambda=1+O(\sqrt{\epsilon})$ if $\epsilon$ is small enough. The running time is $O\big(\frac{d^2}{\epsilon\beta_\epsilon}\log\frac{d}{\beta_\epsilon}+\frac{d}{\epsilon^4}\big)$.
\end{theorem}
\begin{remark}
If the dimensionality $d$ is too high, the random projection based technique {\em Johnson-Lindenstrauss (JL) transform}~\cite{dasgupta2003elementary} can be used to approximately preserve the radius of enclosing ball~\cite{DBLP:conf/compgeom/AgarwalHY07,DBLP:conf/cccg/KerberR15,DBLP:conf/compgeom/Sheehy14}. However, it is not very useful for reducing the time complexity of Algorithm~\ref{alg-meb1}. If we apply JL-transform on the sampled $\Theta(\frac{d}{\beta_\epsilon}\log \frac{d}{\beta_\epsilon})$ points in Step 1, the JL-transform step itself already takes $\Omega(\frac{d^2}{\beta_\epsilon}\log \frac{d}{\beta_\epsilon})$ time (our second algorithm in Section~\ref{sec-sample2} has the time complexity linear in $d$). 
\end{remark}

Before proving Theorem~\ref{the-meb1}, we prove the following lemma first.
\begin{lemma}
\label{lem-net}
Let $S$ be a set of $\Theta(\frac{d}{\beta_\epsilon}\log \frac{d}{\beta_\epsilon})$ points sampled randomly and independently from a given point set $P\subset\mathbb{R}^d$, and $B$ be any ball covering $S$. Then, with constant probability, $|B\cap P|\geq (1-\beta_\epsilon)|P|$.
\end{lemma}
\begin{proof}
Consider the range space $\Sigma=(P, \Phi)$ where each range $\phi\in \Phi$ is the complement of a ball in the space. In a range space, a subset $Y\subset P$ is a $\beta_\epsilon$-net if for any $\phi\in \Phi$, 
$\frac{|P\cap\phi|}{|P|}\geq \beta_\epsilon\Longrightarrow Y\cap\phi\neq\emptyset$. 
Since $|S|=\Theta(\frac{d}{\beta_\epsilon}\log \frac{d}{\beta_\epsilon})$, we know that $S$ is a $\beta_\epsilon$-net of $P$ with constant probability~\cite{vapnik2015uniform,haussler1987}. Thus, if $|B\cap P|< (1-\beta_\epsilon)|P|$, {\em i.e.,} $|P\setminus B|> \beta_\epsilon |P|$, we have $S\cap \big(P\setminus B\big)\neq\emptyset$. This contradicts to  the fact that $S$ is covered by $B$. Consequently, $|B\cap P|\geq (1-\beta_\epsilon)|P|$.
\qed\end{proof}

\begin{proof}(\textbf{of Theorem~\ref{the-meb1}})
Denote by $o$ the center of $MEB(P)$. Since $S\subset P$ and $\mathbb{B}(c, r)$ is a $(1+\epsilon)$-approximate MEB of $S$, we know that $r\leq (1+\epsilon)Rad(P)$. Moreover, Lemma~\ref{lem-net} implies that $|\mathbb{B}(c, r)\cap P|\geq (1-\beta_\epsilon)|P|$ with constant probability. Suppose it is true and let $P'=\mathbb{B}(c, r)\cap P$. Then, we have the distance
\begin{eqnarray}
||c-o||\leq (2+\sqrt{2})\sqrt{\epsilon}  Rad(P) \label{for-the-meb1-1}
\end{eqnarray}
via Theorem~\ref{the-stable}. For simplicity, we use $x$ to denote $(2+\sqrt{2})\sqrt{\epsilon}$. The inequality (\ref{for-the-meb1-1}) implies that the point set $P$ is covered by the ball $\mathbb{B}(c, (1+x)Rad(P))$. Note that we cannot directly return $\mathbb{B}(c, (1+x)Rad(P))$ as the final result, since we do not know the value of $Rad(P)$. Thus, we have to estimate the radius $(1+x)Rad(P)$.

Since $P'$ is covered by $\mathbb{B}(c, r)$ and $|P'|\geq (1-\beta_\epsilon)|P|$, $r$ should be at least $(1-\epsilon)Rad(P)$ due to Definition~\ref{def-stable}. Hence, we have
\begin{eqnarray}
\frac{1+x}{1-\epsilon}r\geq (1+x)Rad(P).
\end{eqnarray}
That is, $P$ is covered by the ball $\mathbb{B}(c, \frac{1+x}{1-\epsilon}r)$. Moreover, the radius 
\begin{eqnarray}
 \frac{1+x}{1-\epsilon}r\leq  \frac{1+x}{1-\epsilon}(1+\epsilon)Rad(P).  
  \end{eqnarray}
This means that  ball $\mathbb{B}(c, \frac{1+x}{1-\epsilon}r)$ is a $\lambda$-approximate MEB of $P$, where
\begin{eqnarray}
\lambda&=&(1+\epsilon) \frac{1+x}{1-\epsilon}=\frac{\big(1+(2+\sqrt{2})\sqrt{\epsilon}\big)(1+\epsilon)}{1-\epsilon} 
\end{eqnarray}
and $\lambda=1+O(\sqrt{\epsilon})$ if $\epsilon$ is small enough. 
If we use the core-set based algorithm~\cite{badoiu2003smaller} to compute $\mathbb{B}(c, r)$ (see Remark~\ref{rem-newbc}), the running time of Algorithm~\ref{alg-meb1} is $O\big(\frac{1}{\epsilon}(|S|d+\frac{1}{\epsilon^3}d)\big)=O\big(\frac{d^2}{\epsilon\beta_\epsilon}\log\frac{d}{\beta_\epsilon}+\frac{d}{\epsilon^4}\big)$.  
\qed
\end{proof}
%


\subsection{The Second Sampling Algorithm}
\label{sec-sample2}

To better understand the second sampling algorithm, we briefly overview our idea below.

\begin{wrapfigure}{r}{0.33\textwidth}
\begin{center}
    \includegraphics[width=0.22\textwidth]{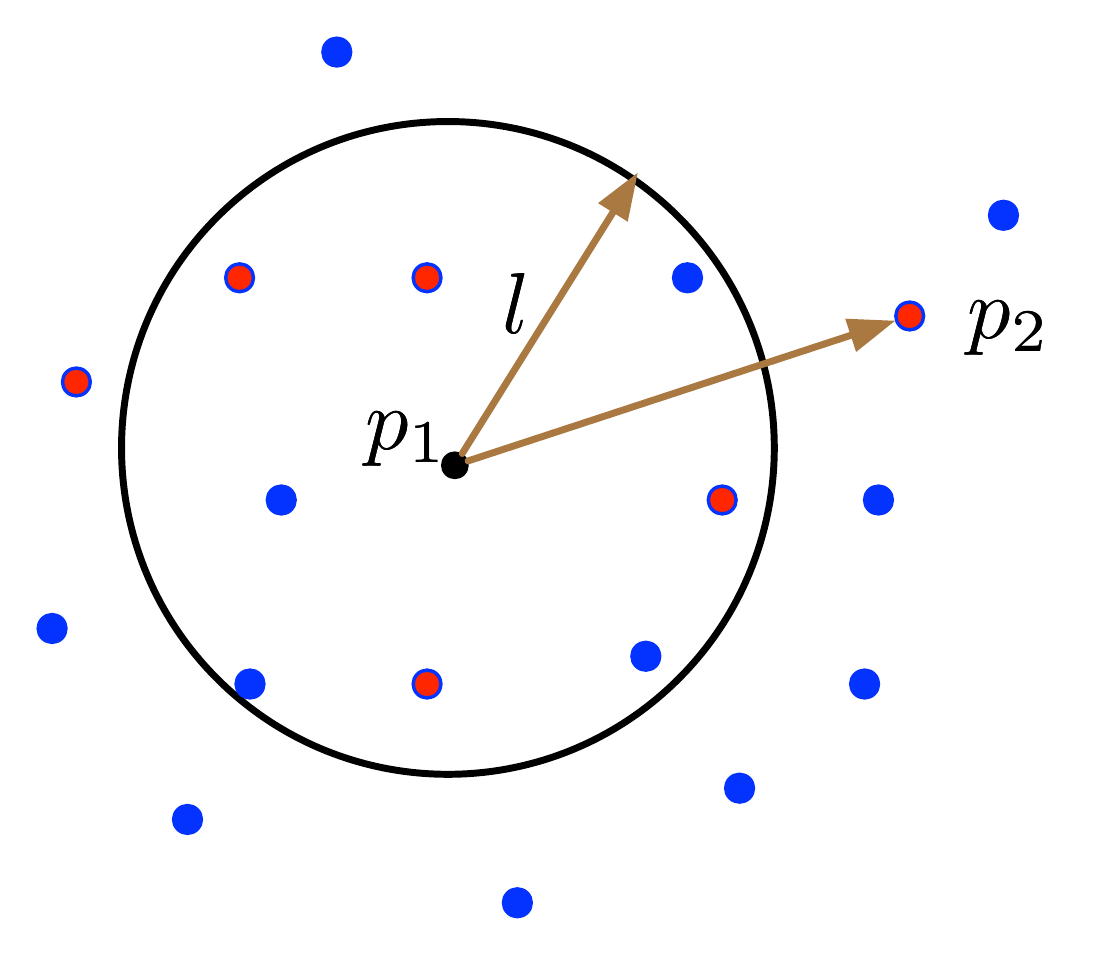}  
    \end{center}
  \vspace{-18pt}
  \caption{An illustration of Lemma~\ref{lem-upper}; the red points are the set $Q$ of sampled points.}     
   \label{fig-lem4}
     \vspace{-18pt}
\end{wrapfigure}

\vspace{0.05in}
\noindent\textbf{High level idea:}  Recall our remark below Theorem~\ref{the-newbc} in Section~\ref{sec-newanalysis}. If we know the value of $(1+\epsilon)Rad(P)$, we can perform almost the same core-set construction procedure described in Theorem~\ref{the-newbc} to achieve an approximate center of $MEB(P)$, where the only difference is that we add a point with distance at least $(1+\epsilon)Rad(P)$ to $o_i$ in each iteration. In this way, we avoid selecting the farthest point to $o_i$, since this operation will inevitably have a linear time complexity. To implement our strategy in sub-linear time, we need to determine the value of  $(1+\epsilon)Rad(P)$ first. We use Lemma~\ref{lem-upper} to estimate the range of $Rad(P)$, and then perform a binary search on the range to determine the value of  $(1+\epsilon)Rad(P)$ approximately. Based on the stability property, we observe that the core-set construction procedure can serve as an ``oracle'' to help us to guess the value of  $(1+\epsilon)Rad(P)$ (see Algorithm~\ref{alg-meb2}). Let $h>0$ be a candidate.  We add a point with distance at least $h$ to $o_i$ in each iteration. We prove that the procedure cannot continue more than $z$ iterations if $h\geq (1+\epsilon)Rad(P)$, and will continue more than $z$ iterations with constant probability if $h<(1-\epsilon)Rad(P)$, where $z=\frac{2}{(1-s)\epsilon}$ is the size of core-set described in Theorem~\ref{the-newbc}. Also, during the procedure of core-set construction, we add the points to the core-set via random sampling, rather than a deterministic way. As a consequence, we obtain our second sub-linear time algorithm where the 
final result is presented in Theorem~\ref{the-sample2}.
\vspace{0.05in}

\begin{lemma}
\label{lem-upper}
Let $P$ be a $\beta_\epsilon$-stable instance of MEB problem. Given a parameter $\eta\in (0,1)$, one selects an arbitrary point $p_1\in P$ and takes a  sample $Q\subset P$ with $|Q|=\frac{1}{\beta_\epsilon}\log\frac{1}{\eta}$ uniformly at random. Let $p_2$ be the point farthest to $p_1$ from $Q$. Then, with probability $1-\eta$, 
\begin{eqnarray}
Rad(P)\in [\frac{1}{2}||p_1-p_2||, \frac{1}{1-\epsilon}||p_1-p_2||].
\end{eqnarray}
\end{lemma}
\begin{proof}

First, the lower bound of $Rad(P)$ is obvious since $||p_1-p_2||$ is always no larger than $2Rad(P)$. Then, we consider the upper bound. Let $\mathbb{B}(p_1, l)$ be the ball covering exactly $(1-\beta_\epsilon)n$ points of $P$, and thus $l\geq (1-\epsilon)Rad(P)$ according to Definition~\ref{def-stable}. To complete our proof, we also need the following folklore lemma presented in~\cite{DX14}.

%

\begin{lemma} {\cite{DX14}}
\label{cla-sampling}
Let $N$ be a set of elements, and $N'$ be a subset of $N$ with
size $\left|N'\right|=\alpha\left|N\right|$ for some $\alpha\in(0,1)$. Given $\eta\in (0,1)$, if one randomly samples $\frac{\ln1/\eta}{\ln1/(1-\alpha)}\leq\frac{1}{\alpha}\ln\frac{1}{\eta}$ elements from $N$, then with probability at least $1-\eta$, the sample contains at least one element of $N'$.
\end{lemma}

In Lemma~\ref{cla-sampling}, let $N$  and $N'$ be the point set $P$ and the subset $P\setminus \mathbb{B}(p_1, l)$, respectively. We know that $Q$ contains at least one point from $N'$ according to Lemma~\ref{cla-sampling} (by setting $\alpha=\beta_\epsilon$). Namely, $Q$ contains at least one point outside $\mathbb{B}(p_1, l)$. See Figure~\ref{fig-lem4}. As a consequence, we have $||p_1-p_2||\geq l\geq (1-\epsilon)Rad(P)$, {\em i.e.}, $Rad(P)\leq \frac{1}{1-\epsilon}||p_1-p_2||$.
\qed\end{proof}

Note that Lemma~\ref{lem-upper} immediately implies the following result.

\begin{theorem}
\label{the-sample2-easy}
In Lemma~\ref{lem-upper}, the ball $\mathbb{B}(p_1, \frac{2}{1-\epsilon}||p_1-p_2||)$ is a $\frac{4}{1-\epsilon}$-approximate MEB of $P$, with probability $1-\eta$.
\end{theorem}
\begin{proof}
From the upper bound in Lemma~\ref{lem-upper}, we know that $\frac{2}{1-\epsilon}||p_1-p_2||\geq 2Rad(P)$. Since $||p_1-p||\leq 2 Rad(P)$ for any $p\in P$,  the ball $\mathbb{B}(p_1, \frac{2}{1-\epsilon}||p_1-p_2||)$ covers the whole point set $P$. From the lower bound in Lemma~\ref{lem-upper}, we know that $\frac{2}{1-\epsilon}||p_1-p_2||\leq \frac{4}{1-\epsilon}Rad(P)$. Therefore, it is a $\frac{4}{1-\epsilon}$-approximate MEB of $P$.
\qed\end{proof}
\vspace{-0.05in}

Since $|Q|=\frac{1}{\beta_\epsilon}\log\frac{1}{\eta}$ in Lemma~\ref{lem-upper}, Theorem~\ref{the-sample2-easy} indicates that we can easily obtain a $\frac{4}{1-\epsilon}$-approximate MEB of $P$ in $O(\frac{1}{\beta_\epsilon}(\log\frac{1}{\eta}) d)$ time. We further show our second sampling algorithm (Algorithm~\ref{alg-meb3}) that achieves a lower approximation ratio. Algorithm~\ref{alg-meb2} serves as a subroutine in Algorithm~\ref{alg-meb3}. In Algorithm~\ref{alg-meb2}, we simply set $z=\frac{3}{\epsilon}$ with $s=1/3$ as described in Theorem~\ref{the-newbc}; we compute $o_i$ having distance less than $s\frac{\epsilon}{1+\epsilon}Rad(T)$ to the center of $MEB(T)$ in Step 2(1).

\begin{algorithm}
   \caption{MEB Algorithm \Rmnum{2}}
   \label{alg-meb3}
\begin{algorithmic}[1]
\REQUIRE A $\beta_\epsilon$-stable instance $P$ of MEB problem in $\mathbb{R}^d$; a parameter $\eta_0\in (0,1)$ and a positive integer $z=\frac{3}{\epsilon}$; the interval $[a, b]$ for $Rad(P)$ obtained by Lemma~\ref{lem-upper}.
\STATE Among the set $\{(1-\epsilon)a, (1+\epsilon)(1-\epsilon)a, \cdots, (1+\epsilon)^w (1-\epsilon)a=(1+\epsilon)b\}$ where $w=\lceil \log_{1+\epsilon}\frac{2}{(1-\epsilon)^2}\rceil+1=O(\frac{1}{\epsilon}\log \frac{1}{1-\epsilon})$, perform binary search for the value $h$ by using Algorithm~\ref{alg-meb2} with $\eta=\frac{\eta_0}{2\log w}$. 
\STATE Suppose that Algorithm~\ref{alg-meb2} returns ``no'' when $h=(1+\epsilon)^{i_0} (1-\epsilon)a$ and returns ``yes'' when $h=(1+\epsilon)^{i_0+1} (1-\epsilon)a$.
\STATE Run Algorithm~\ref{alg-meb2} again with $h=(1+\epsilon)^{i_0+2}a$ and $\eta=\eta_0/2$; let $\tilde{o}$ be the resulting ball center of $T$ when the loop stops.
\STATE Return the ball $\mathbb{B}(\tilde{o}, r)$, where $r=\frac{1+(4+4\sqrt{2})\sqrt{\frac{\epsilon}{1-\epsilon}}}{1+\epsilon}h$.
\end{algorithmic}
\end{algorithm}
\vspace{-0.2in}

\begin{theorem}
\label{the-sample2}
With probability $1-\eta_0$, Algorithm~\ref{alg-meb3} returns a $\lambda$-approximate MEB of $P$, where 
\begin{eqnarray}
\lambda=\frac{(1+x_1)(1+x_2)}{1+\epsilon} \text{\hspace{0.1in} with \hspace{0.1in}} x_1=O\big(\frac{\epsilon}{1-\epsilon}\big), x_2=O\big(\sqrt{\frac{\epsilon}{1-\epsilon}}\big),
\end{eqnarray}
and $\lambda=1+O(\sqrt{\epsilon})$ if $\epsilon$ is small enough. The running time is $\tilde{O}\big((\frac{1}{\epsilon\beta_\epsilon}+\frac{1}{\epsilon^4})d\big)$, where  $\tilde{O}(f)=O(f\cdot polylog(\frac{1}{\epsilon}, \frac{1}{1-\epsilon}, \frac{1}{\eta_0}))$. 
\end{theorem}

Before proving Theorem~\ref{the-sample2}, we provide Lemma~\ref{lem-sample2} first.

\begin{algorithm}
   \caption{Oracle for testing $h$}
   \label{alg-meb2}
\begin{algorithmic}[1]
\REQUIRE A $\beta_\epsilon$-stable instance $P$ of MEB problem in $\mathbb{R}^d$; a parameter  $\eta\in (0,1)$, $h>0$, and a positive integer $z=\frac{3}{\epsilon}$.
\STATE Initially, arbitrarily select a point $p\in P$ and let $T=\{p\}$.
\STATE $i=1$; repeat the following steps:
\begin{enumerate}[(1)]
\item Compute an approximate MEB of $T$ and let the ball center be $o_i$.
\item Sample a set $Q\subset P$ with $|Q|=\frac{1}{\beta_\epsilon}\log\frac{z}{\eta}$ uniformly at random.
\item Select the point $q\in Q$ that is farthest to $o_i$, and add it to $T$.
\item If $||q-o_i||< h$, stop the loop and output ``yes''.
\item $i=i+1$; if $i>z$, stop the loop and output ``no''.
\end{enumerate}
\end{algorithmic}
\end{algorithm}

\begin{lemma}
\label{lem-sample2}
If $h\geq (1+\epsilon)Rad(P)$, Algorithm~\ref{alg-meb2} returns ``yes''; else if $h<(1-\epsilon)Rad(P)$, Algorithm~\ref{alg-meb2} returns ``no'' with probability at least $1-\eta$.
\end{lemma}
\begin{proof}
First, we assume that $h\geq (1+\epsilon)Rad(P)$. 
Recall the remark following Theorem~\ref{the-newbc}. 
If we always add a point $q$ with distance at least $h\geq (1+\epsilon)Rad(P)$ to $o_i$, the loop 2(1)-(5) cannot continue more than $z$ iterations, {\em i.e.}, Algorithm~\ref{alg-meb2} will return ``yes''.

Now, we consider the case $h<(1-\epsilon)Rad(P)$. Similar to the proof of Lemma~\ref{lem-upper}, we consider the ball $\mathbb{B}(o_i, l)$ covering exactly $(1-\beta_\epsilon)n$ points of $P$. We know that $l\geq (1-\epsilon)Rad(P)>h$ according to Definition~\ref{def-stable}. Also, with probability $1-\eta/z$, the sample $Q$ contains at least one point outside $\mathbb{B}(o_i, l)$ due to Lemma~\ref{cla-sampling}. By taking the union bound, with probability $(1-\eta/z)^z\geq 1-\eta$, $||q-o_i||$ is always larger than $h$ and eventually Algorithm~\ref{alg-meb2} will return ``no''.
\qed\end{proof}

\begin{proof}[\textbf{of Theorem~\ref{the-sample2}}]
Since Algorithm~\ref{alg-meb2} returns ``no'' when $h=(1+\epsilon)^{i_0} (1-\epsilon)a$ and returns ``yes'' when $h=(1+\epsilon)^{i_0+1} (1-\epsilon)a$, we know that 
\begin{eqnarray}
(1+\epsilon)^{i_0} (1-\epsilon)a&<&(1+\epsilon)Rad(P); \label{for-the-sample2-2}\\
(1+\epsilon)^{i_0+1} (1-\epsilon)a&\geq& (1-\epsilon)Rad(P),\label{for-the-sample2-1}
\end{eqnarray}
from Lemma~\ref{lem-sample2}. The above inequalities together imply that
\begin{eqnarray}
\frac{(1+\epsilon)^3}{1-\epsilon}Rad(P)>(1+\epsilon)^{i_0+2}a\geq (1+\epsilon)Rad(P).\label{for-the-sample2-3}
\end{eqnarray}
Thus, when running Algorithm~\ref{alg-meb2} with $h=(1+\epsilon)^{i_0+2}a$ in Step 3, the algorithm returns ``yes'' (by the right hand-side of (\ref{for-the-sample2-3})). Then, consider the ball $\mathbb{B}(\tilde{o}, h)$. We claim that $|P\setminus\mathbb{B}(\tilde{o}, h)|<\beta_\epsilon n$. Otherwise, the sample $Q$ contains at least one point outside $\mathbb{B}(\tilde{o}, h)$ with probability $1-\eta/z$ in Step 2(2) of Algorithm~\ref{alg-meb2}, {\em i.e.,} the loop will continue. Thus, it  contradicts to the fact that the algorithm returns ``yes''. Let 
$P'=P\cap\mathbb{B}(\tilde{o}, h)$, and then $|P'|\geq(1-\beta_\epsilon)n$. Moreover, the left hand-side of (\ref{for-the-sample2-3}) indicates that
\begin{eqnarray}
h=(1+\epsilon)^{i_0+2}a\leq(1+\frac{8\epsilon}{1-\epsilon})Rad(P).\label{for-the-sample2-4}
\end{eqnarray}
Now, we can apply Theorem~\ref{the-stable}, where the only difference is that we replace the ``$\epsilon$'' by ``$\frac{8\epsilon}{1-\epsilon}$'' in the theorem. Let $o$ be the center of $MEB(P)$. Consequently, we have
\begin{eqnarray}
||\tilde{o}-o||\leq (4+4\sqrt{2})\sqrt{\frac{\epsilon}{1-\epsilon}} Rad(P).\label{for-the-sample2-5}
\end{eqnarray}

For simplicity, we let $x_1=\frac{8\epsilon}{1-\epsilon}$ and $x_2=(4+4\sqrt{2})\sqrt{\frac{\epsilon}{1-\epsilon}}$. Hence, $h\leq (1+x_1)Rad(P)$ and $||\tilde{o}-o||\leq x_2Rad(P)$ in (\ref{for-the-sample2-4}) and (\ref{for-the-sample2-5}). From (\ref{for-the-sample2-5}), we know that $P\subset\mathbb{B}(\tilde{o}, (1+x_2)Rad(P))$.  From the right hand-side of (\ref{for-the-sample2-3}), we know that $(1+x_2)Rad(P)\leq\frac{1+x_2}{1+\epsilon}h$. Thus, we have 
\begin{eqnarray}
P\subset\mathbb{B}\Big(\tilde{o}, \frac{1+x_2}{1+\epsilon}h\Big) \label{for-the-sample2-6},
\end{eqnarray}
where $\frac{1+x_2}{1+\epsilon}h=\frac{1+(4+4\sqrt{2})\sqrt{\frac{\epsilon}{1-\epsilon}}}{1+\epsilon}h$. Also, the radius
\begin{eqnarray}
\frac{1+x_2}{1+\epsilon}h&\leq&\frac{(1+x_2)(1+x_1)}{1+\epsilon}Rad(P)=\lambda Rad(P).
\end{eqnarray}
This means that  $\mathbb{B}\Big(\tilde{o}, \frac{1+x_2}{1+\epsilon}h\Big)$ is a $\lambda$-approximate MEB of $P$; $\lambda=1+O(\sqrt{\epsilon})$ if $\epsilon$ is small enough.

\vspace{0.05in}
\textbf{Success probability.} The success probability of Algorithm~\ref{alg-meb2} is $1-\eta$. In Algorithm~\ref{alg-meb3}, we set $\eta=\frac{\eta_0}{2\log w}$ in Step 1 and $\eta=\eta_0/2$ in Step 3, respectively. We take the union bound and the success probability of Algorithm~\ref{alg-meb3} is $(1-\frac{\eta_0}{2\log w})^{\log w} (1-\eta_0/2)>1-\eta_0$.

\vspace{0.05in}
\textbf{Running time.} As the subroutine, Algorithm~\ref{alg-meb2} runs in $O(z(\frac{1}{\beta_\epsilon}(\log\frac{z}{\eta}) d+\frac{1}{\epsilon^3}d))$ time; Algorithm~\ref{alg-meb3} calls the subroutine $O\big(\log(\frac{1}{\epsilon}\log \frac{1}{1-\epsilon})\big)$ times. Note that $z=O(\frac{1}{\epsilon})$. Thus, the total running time  is $\tilde{O}\big((\frac{1}{\epsilon\beta_\epsilon}+\frac{1}{\epsilon^4})d\big)$. 
\qed
\end{proof}

\section{Sub-linear Time Algorithm for MEB with Outliers (sketch)}
\label{sec-outlier-stable2}
\begin{wrapfigure}{r}{0.25\textwidth}
 \vspace{-0.25in}
\begin{center}
    \includegraphics[width=0.22\textwidth]{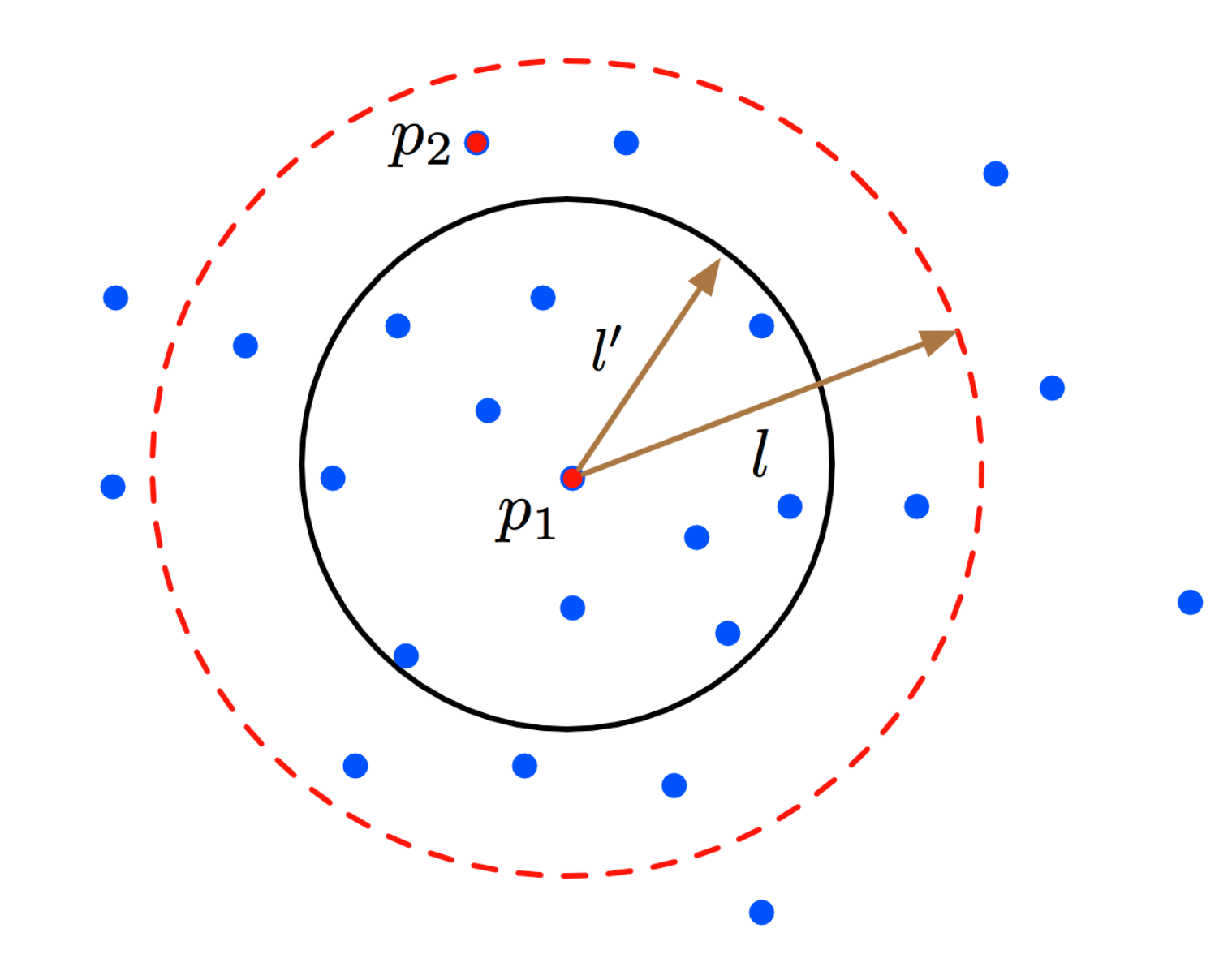}  
    \end{center}
  \vspace{-0.35in}
  \caption{}     
   \label{fig-lem7-2}
     \vspace{-0.3in}
\end{wrapfigure}
 Our result for MEB with outliers is an extension of Theorem~\ref{the-sample2-easy}, but needs a more complicated analysis. A key step is to estimate the range of $Rad(P_{opt})$. In Lemma~\ref{lem-upper}, we can estimate the range of $Rad(P)$ via a simple sampling procedure. However, this idea cannot be applied to the case with outliers, since the farthest sampled point $p_2$ could be an outlier. To address this issue, we imagine two balls centered at $p_1$ with two carefully chosen radii $l$ and $l'$ (see Figure~\ref{fig-lem7-2}). Recall in the proof of Lemma~\ref{lem-upper}, we only consider one ball $\mathbb{B}(p_1, l)$ as in Figure~\ref{fig-lem4}. Intuitively, these two balls, $\mathbb{B}(p_1, l)$ and $\mathbb{B}(p_1, l')$, guarantee a large enough gap such that there exists at least one sampled point, say $p_2$, falling in the ring between the two spheres.  Moreover, together with the stability property described in Definition~\ref{def-outlier-stable}, we show that the distance $||p_1-p_2||$ can provide a range of $Rad(P_{opt})$. \textbf{Due to the space limit, we leave the details to Section~\ref{sec-outlier-stable}.}

\newpage
\bibliographystyle{abbrv}

\bibliography{stability}


\section{Proof of Theorem~\ref{the-newbc}\cite{DBLP:journals/corr/abs-2004-10090}}
\label{sec-proof-newbc}
To ensure the expected improvement in each iteration of the algorithm of~\cite{badoiu2003smaller}, they showed that the following two inequalities hold if the algorithm always selects the farthest point to the current center of $MEB(T)$:
\begin{eqnarray}
r_{i+1}  \geq  (1+\epsilon)Rad(P)-L_i; \hspace{0.2in} r_{i+1} \geq  \sqrt{r^2_i+L^2_i},\label{for-bc2}
 \end{eqnarray}
where $r_i$ and $r_{i+1}$ are the radii of $MEB(T)$ in the $i$-th and $(i+1)$-th iterations, respectively, and $L_i$ is the shifting distance of the center of $MEB(T)$ from the $i$-th to $(i+1)$-th iteration.

However, we often compute only an approximate $MEB(T)$ in each iteration. In the $i$-th iteration, we let $c_i$ and $o_i$ denote the centers of the exact and the approximate $MEB(T)$, 
respectively. Suppose that $||c_i-o_i||\leq \xi r_i$, where $\xi\in (0,\frac{\epsilon}{1+\epsilon})$ (we will see why  this bound is needed later). Note that we only compute $o_i$ rather than $c_i$ in each iteration. As a consequence, we can only select the farthest point (say $q$) to $o_i$. If $||q-o_i||\leq (1+\epsilon)Rad(P)$,  we are done and a $(1+\epsilon)$-approximation of MEB is already obtained. Otherwise, we have
\begin{eqnarray}
(1+\epsilon)Rad(P)&<& ||q-o_i||\nonumber\\
&\leq& ||q-c_{i+1}||+||c_{i+1}-c_i||+||c_i-o_i||\nonumber\\
&\leq& r_{i+1}+L_i+\xi r_i \label{for-bc-1}
\end{eqnarray}
by  the triangle inequality. In other words, we should replace the first inequality of (\ref{for-bc2}) by $r_{i+1} > (1+\epsilon)Rad(P)-L_i-\xi r_i$. Also, the second inequality of (\ref{for-bc2}) still holds since it depends only on the property of the exact MEB (see Lemma 2.1 in~\cite{badoiu2003smaller}). Thus,  we have 
\begin{eqnarray}
r_{i+1}\geq \max\Big\{\sqrt{r^2_i+L^2_i}, (1+\epsilon)Rad(P)-L_i-\xi r_i\Big\}.\label{for-bc4}
\end{eqnarray}

Similar to the analysis in~\cite{badoiu2003smaller}, we let $\lambda_i=\frac{r_i}{(1+\epsilon)Rad(P)}$. Because $r_i$ is the radius of $MEB(T)$ and $T\subset P$,  we know $r_i\leq Rad(P)$ and then $\lambda_i\leq1/(1+\epsilon)$. By simple calculation, we know that when $L_i=\frac{\big((1+\epsilon)Rad(P)-\xi r_i\big)^2-r^2_i}{2\big((1+\epsilon)Rad(P)-\xi r_i\big)}$ the lower bound of $r_{i+1}$ in (\ref{for-bc4}) achieves the minimum value. Plugging this value of $L_i$ into (\ref{for-bc4}), we have
\begin{eqnarray}
\lambda^2_{i+1}\geq \lambda^2_i+\frac{\big((1-\xi\lambda_i)^2-\lambda^2_i\big)^2}{4(1-\xi\lambda_i)^2}.\label{for-bc5}
\end{eqnarray}
To simplify  inequality (\ref{for-bc5}), we consider the function $g(x)=\frac{(1-x)^2-\lambda^2_i}{1-x}$, where $0<x<\xi$. Its derivative $g'(x)=-1-\frac{\lambda^2_i}{(1-x)^2}$ is always negative, thus we have
\begin{eqnarray}
g(x)\geq g(\xi)=\frac{(1-\xi)^2-\lambda^2_i}{1-\xi}. \label{for-bc2-1}
\end{eqnarray}
Because $\xi<\frac{\epsilon}{1+\epsilon}$ and $\lambda_i\leq 1/(1+\epsilon)$, we know that  the right-hand side of (\ref{for-bc2-1}) is always non-negative. Using (\ref{for-bc2-1}), inequality (\ref{for-bc5}) can be simplified to 
\begin{eqnarray}
\lambda^2_{i+1}&\geq& \lambda^2_i+\frac{1}{4}\big(g(\xi)\big)^2\nonumber\\
&=&\lambda^2_i+\frac{\big((1-\xi)^2-\lambda^2_i\big)^2}{4(1-\xi)^2}.\label{for-bc2-2}
\end{eqnarray}
(\ref{for-bc2-2}) can be further rewritten as 
\begin{eqnarray}
\Big(\frac{\lambda_{i+1}}{1-\xi}\Big)^2&\geq&\frac{1}{4}\Big(1+(\frac{\lambda_{i}}{1-\xi})^2\Big)^2 \nonumber\\
\Longrightarrow  \frac{\lambda_{i+1}}{1-\xi}&\geq&\frac{1}{2}\Big(1+(\frac{\lambda_{i}}{1-\xi})^2\Big).\label{for-bc2-3}
\end{eqnarray}

Now, we can apply a similar transformation of $\lambda_i$ which was used in~\cite{badoiu2003smaller}. Let $\gamma_i=\frac{1}{1-\frac{\lambda_i}{1-\xi}}$.  We know $\gamma_i>1$ (note $0\leq\lambda_i\leq\frac{1}{1+\epsilon}$ and $\xi<\frac{\epsilon}{1+\epsilon}$). Then, (\ref{for-bc2-3}) implies that 
\begin{eqnarray}
\gamma_{i+1}&\geq&\frac{\gamma_i}{1-\frac{1}{2\gamma_i}}\nonumber\\
&=&\gamma_i\big(1+\frac{1}{2\gamma_i}+(\frac{1}{2\gamma_i})^2+\cdots\big)\nonumber\\
&>&\gamma_i+\frac{1}{2}, \label{for-bc2-4}
\end{eqnarray}
where the equation comes from the fact that $\gamma_i>1$ and thus $\frac{1}{2\gamma_i}\in(0,\frac{1}{2})$. Note that $\lambda_0=0$ and thus $\gamma_0=1$. As a consequence, we have $\gamma_i>1+\frac{i}{2}$. In addition, since $\lambda_i\leq\frac{1}{1+\epsilon}$, that is, $\gamma_i\leq\frac{1}{1-\frac{1}{(1+\epsilon)(1-\xi)}}$, we have
\begin{eqnarray}
i< \frac{2}{\epsilon-\xi-\epsilon\xi}=\frac{2}{(1-\frac{1+\epsilon}{\epsilon}\xi)\epsilon}.\label{for-bc2-5}
\end{eqnarray}

Consequently, we obtain the theorem.

\section{Lemma 2.2 in \cite{BHI}}
\label{sec-bhi}

\begin{lemma}[\cite{BHI}]
Let $\mathbb{B}(c, r)$ be a minimum enclosing ball of a point set $P\subset\mathbb{R}^d$, then any closed half-space that contains $c$, must also contain at least a point from $P$ that is at distance $r$ from $c$.
\end{lemma}

\section{Sub-linear Time Algorithm for MEB with Outliers}
\label{sec-outlier-stable}
 %
The result in this section is an extension of Theorem~\ref{the-sample2-easy}, but needs a more complicated analysis. A key step is to estimate the range of $Rad(P_{opt})$. In Lemma~\ref{lem-upper}, we can estimate the range via a simple sampling procedure. However, this idea cannot be applied to the case with outliers, since the farthest sampled point $p_2$ could be an outlier. We briefly introduce our idea below. 

\textbf{High level idea:}  To estimate the range of $Rad(P_{opt})$, we imagine two balls centered at $p_1$ with two appropriate radii (see Figure~\ref{fig-lem7}). Recall in the proof of Lemma~\ref{lem-upper}, we only consider one ball $\mathbb{B}(p_1, l)$ as in Figure~\ref{fig-lem4}. Intuitively, these two balls in Figure~\ref{fig-lem7} guarantee a large enough gap such that there exists at least one sampled point, say $p_2$, falling in the ring between the two spheres.  Moreover, together with the stability property described in Definition~\ref{def-outlier-stable}, we show that the distance $||p_1-p_2||$ can provide a range of $Rad(P_{opt})$ in Lemma~\ref{lem-outlier-estimate}. 


\begin{lemma}
\label{lem-outlier-estimate}
Let $(P,\gamma)$ be a $\beta_\epsilon$-stable instance of MEB with outliers, and $p_1$ be a point randomly selected from $P$. Let $Q$ be a random sample from $P$ with 
size $|Q|=O\big(\max\{\frac{1}{\beta_\epsilon},\frac{1}{\gamma}\}\times\frac{(2\gamma+\beta_\epsilon)^2}{\beta_\epsilon^2}\log\frac{1}{\eta}\big)$ for a given $\eta\in (0,1)$.
Then, if $p_2$ is the $t$-th farthest point to $p_1$ in $Q$, where $t=\frac{2\gamma+2\beta_\epsilon}{2\gamma+\beta_\epsilon}\gamma |Q|+1$, the following holds with probability $(1-\eta)(1-\gamma)$, 
\begin{eqnarray}
Rad(P_{opt})\in [\frac{1}{2}||p_1-p_2||, \frac{1}{1-\epsilon}||p_1-p_2||].
\end{eqnarray}
\end{lemma}
\begin{proof}

\begin{figure}[]
\vspace{-0.1in}
   \centering
  \includegraphics[height=1.6in]{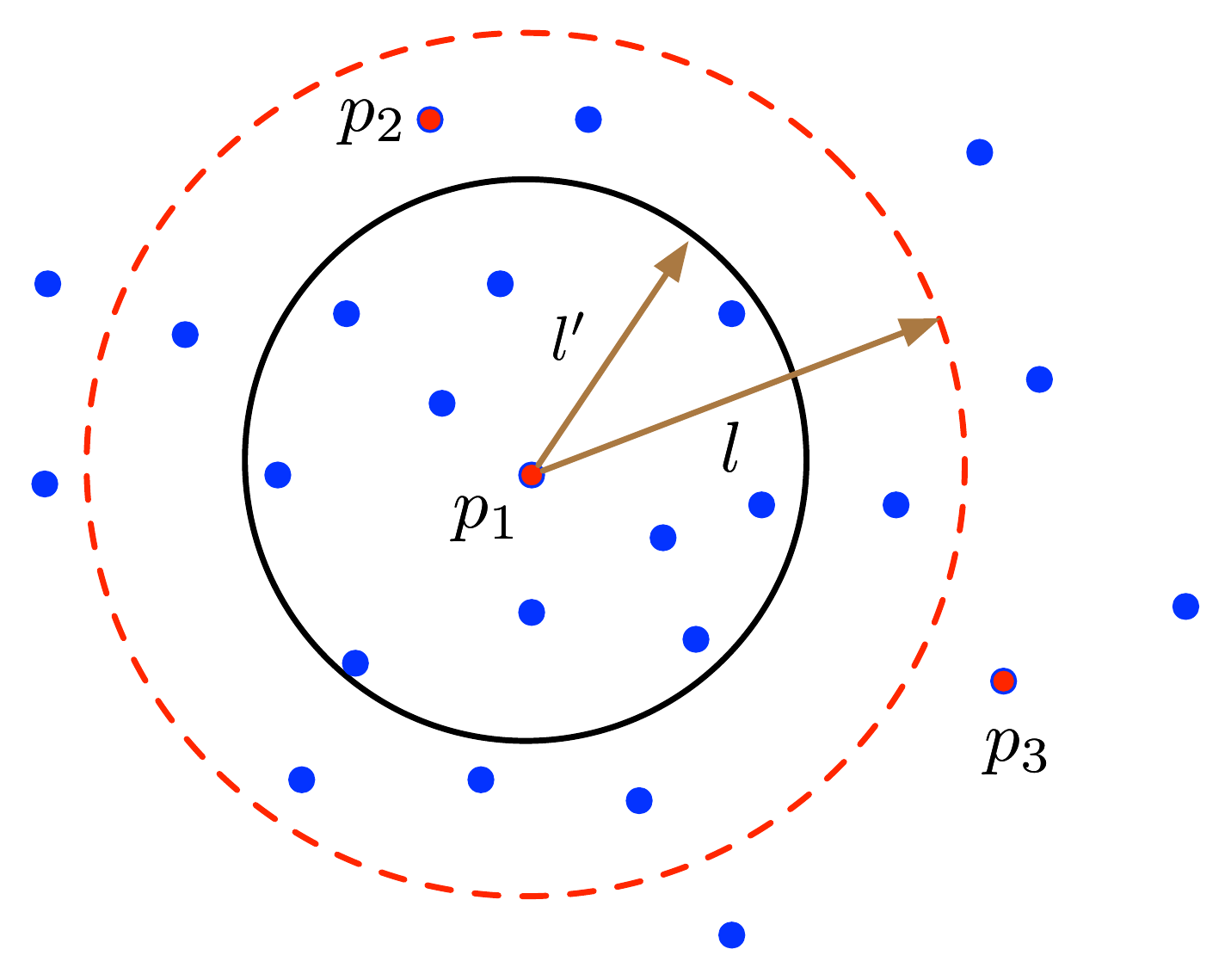}
  \vspace{-0.1in}
      \caption{An illustration of Lemma~\ref{lem-outlier-estimate}.}
  \label{fig-lem7}
  \vspace{-0.1in}
\end{figure}

First, we assume that $p_1\in P_{opt}$ (note that this happens with probability $1-\gamma$). We consider two balls $\mathbb{B}(p_1, l')$ and $\mathbb{B}(p_1, l)$ such that 
\begin{eqnarray}
|P\cap\mathbb{B}(p_1, l')|&=&(1-\gamma-\beta_\epsilon)n; \label{for-outlier-e1}\\
|P\cap\mathbb{B}(p_1, l)|&=&(1-\gamma)n.\label{for-outlier-e2}
\end{eqnarray}
That is, $\mathbb{B}(p_1, l)$ contains $\beta_\epsilon n$ more points than $\mathbb{B}(p_1, l')$ from $P$ (see Figure~\ref{fig-lem7}). Further, we define two subsets $A=P\setminus \mathbb{B}(p_1, l)$ and $B=P\cap(\mathbb{B}(p_1, l)\setminus \mathbb{B}(p_1, l'))$. Therefore, $|A|=\gamma n$ and $|B|=\beta_\epsilon n$. 

Now, suppose that we randomly sample $m$ points $Q$ from $P$, where the value of $m$ will be determined later. Let $\{x_i\mid 1\leq i\leq m\}$ be $m$ independent random variables with $x_i=1$ if the $i$-th sampled point belongs to $A$, and $x_i=0$ otherwise. Thus, $E[x_i]=\gamma$ for each $i$. Let $\sigma$ be a small parameter in $(0,1)$. By using the Chernoff bound, we have  $\textbf{Pr}\Big(\sum^m_{i=1}x_i\notin (1\pm\sigma)\gamma m\Big)\leq e^{-O(\sigma^2 m\gamma)}$. That is,
\begin{eqnarray}
\textbf{Pr}\Big(|Q\cap A|\in (1\pm\sigma)\gamma m\Big)\geq 1-e^{-O(\sigma^2 m\gamma)}. \label{for-outlier-e3}
\end{eqnarray}
Similarly, we have
\begin{eqnarray}
\textbf{Pr}\Big(|Q\cap B|\in (1\pm\sigma)\beta_\epsilon m\Big)\geq 1-e^{-O(\sigma^2 m\beta_\epsilon)}. \label{for-outlier-e4}
\end{eqnarray}
Consequently, if $m=O(\max\{\frac{1}{\beta_\epsilon}, \frac{1}{\gamma}\}\times \frac{1}{\sigma^2}\log\frac{2}{\eta})$, with probability $(1-\frac{\eta}{2})^2>1-\eta$, we have 
\begin{eqnarray}
|Q\cap A|\in (1\pm\sigma)\gamma m \text{\hspace{0.2in } and \hspace{0.2in }} |Q\cap B|\in (1\pm\sigma)\beta_\epsilon m. \label{for-outlier-e5}
\end{eqnarray} 
Therefore, if we rank the points of $Q$ by their distances to $p_1$ decreasingly, we know that at most the top $(1+\sigma)\gamma m$ points belong to $A$, and at least the top $(1-\sigma)(\gamma+\beta_\epsilon)m$ points belong to $A\cup B$. To ensure $(1+\sigma)\gamma m<(1-\sigma)(\gamma+\beta_\epsilon)m$ ({\em i.e.}, there is a gap between $(1+\sigma)\gamma m$ and $(1-\sigma)(\gamma+\beta_\epsilon)m$), we need to set $\sigma<\frac{\beta_\epsilon}{2\gamma+\beta_\epsilon}$ ({\em e.g.}, we can set $\sigma=\frac{1}{2}\frac{\beta_\epsilon}{2\gamma+\beta_\epsilon}$). Then, we pick the $t$-th farthest point to $p_1$ from $Q$, where $t=(1+\sigma)\gamma m+1$, and denote it as $p_2$. As a consequence, $p_2\in B$ with probability $1-\eta$.

Suppose $p_2\in B$ (see Figure~\ref{fig-lem7}). From Definition~\ref{def-outlier-stable}, we have 
\begin{eqnarray}
||p_1-p_2||\geq l'\geq (1-\epsilon)Rad(P_{opt}). \label{for-outlier-e6}
\end{eqnarray}
To obtain the upper bound of $||p_1-p_2||$, we consider two cases: $A\cap P_{opt}=\emptyset$ and $A\cap P_{opt}\neq\emptyset$. For the former case, since $A=P\setminus \mathbb{B}(p_1, l)$ and $|A|=\gamma n$, we know that the whole $P_{opt}$ is covered by $\mathbb{B}(p_1, l)$ and all the points of $P\setminus P_{opt}$ are outside of $\mathbb{B}(p_1, l)$. 
Since $p_2\in B\subset\mathbb{B}(p_1, l)$,  we know $p_2\in P_{opt}$. It implies that $||p_1-p_2||\leq 2Rad(P_{opt})$. For the latter case, let $p_3\in A\cap P_{opt}$ (see Figure~\ref{fig-lem7}). Then we have 
\begin{eqnarray}
||p_1-p_2||\leq l\leq ||p_1-p_3|| \leq 2Rad(P_{opt}).
\end{eqnarray}
Thus, we have $||p_1-p_2||\leq 2Rad(P_{opt})$ for both cases. 

Overall, $p_1\in P_{opt}$ with probability $1-\gamma$ and $p_2\in  B$ with probability $1-\eta$, and thus $Rad(P_{opt})\in [\frac{1}{2}||p_1-p_2||, \frac{1}{1-\epsilon}||p_1-p_2||]$ with probability $(1-\eta)(1-\gamma)$. We set $\sigma=\frac{1}{2}\frac{\beta_\epsilon}{2\gamma+\beta_\epsilon}$, and the sample size $|Q|=m=O(\max\{\frac{1}{\beta_\epsilon}, \frac{1}{\gamma}\}\times \frac{1}{\sigma^2}\log\frac{2}{\eta})=O\big(\max\{\frac{1}{\beta_\epsilon},\frac{1}{\gamma}\}\times\frac{(2\gamma+\beta_\epsilon)^2}{\beta_\epsilon^2}\log\frac{1}{\eta}\big)$.
\qed
\end{proof}

Similar to Theorem~\ref{the-sample2-easy}, we can obtain an approximate solution of MEB with outliers via Lemma~\ref{lem-outlier-estimate}. In the proof of Lemma~\ref{lem-outlier-estimate}, we assume $p_1\in P_{opt}$, and thus $P_{opt}\subset\mathbb{B}(p_1, 2Rad(P_{opt}))$. Moreover, since $Rad(P_{opt})\in [\frac{1}{2}||p_1-p_2||, \frac{1}{1-\epsilon}||p_1-p_2||]$, we know that 
\begin{eqnarray}
 \mathbb{B}(p_1, 2Rad(P_{opt}))&\subset& \mathbb{B}(p_1, \frac{2}{1-\epsilon}||p_1-p_2||);\\
 \frac{2}{1-\epsilon}||p_1-p_2||&\leq& \frac{4}{1-\epsilon}Rad(P_{opt}). 
 \end{eqnarray}
 Note that $p_2$ can be selected from the sample $Q$ in linear time $O(|Q|d)$ by the algorithm in~\cite{blum1973time}. Thus, we have the following result.

\begin{theorem}
\label{the-outlier-sub}
In Lemma~\ref{lem-outlier-estimate}, the ball $\mathbb{B}(p_1, \frac{2}{1-\epsilon}||p_1-p_2||)$ is a $\frac{4}{1-\epsilon}$-approximation of the instance $(P, \gamma)$, with probability $(1-\eta)(1-\gamma)$. The running time for obtaining the ball is $O\big(\max\{\frac{1}{\beta_\epsilon},\frac{1}{\gamma}\}\times\frac{(2\gamma+\beta_\epsilon)^2}{\beta_\epsilon^2}(\log\frac{1}{\eta})d\big)$.
\end{theorem}

\end{document}